\documentclass[letterpaper, 11pt, notitlepage]{article}
\usepackage{times}
\usepackage{amsthm}
\usepackage{amsmath}
\usepackage{amssymb}
\usepackage{textcomp}
\usepackage{graphicx}
\usepackage{color}
\usepackage{verbatim}
\makeatletter
\g@addto@macro\@verbatim\footnotesize
\makeatother
\newtheorem{thm}{Theorem}
\newtheorem{lemma}[thm]{Lemma}
\newtheorem{prop}[thm]{Proposition}
\newtheorem{cor}[thm]{Corollary}
\newenvironment{defn}[1][Definition.]{\begin{trivlist}\item[\hskip \labelsep {\bfseries #1}]}{\end{trivlist}}

\title{There is no 16-Clue Sudoku: Solving the Sudoku Minimum Number of Clues Problem via Hitting Set Enumeration}
\date{August 31, 2013}
\author{Gary McGuire\footnote{School of Mathematical Sciences, University College Dublin, Ireland. E-mail: gary.mcguire@ucd.ie}, \ Bastian Tugemann\footnote{Munich, Germany.}, \ Gilles Civario\footnote{Irish Centre for High-End Computing, Dublin, Ireland.}}

\begin{document}
\baselineskip=15pt

\maketitle

\begin{abstract}
The sudoku minimum number of clues problem is the following question: what is the smallest number of clues that a sudoku puzzle can have?
For several years it had been conjectured that the answer is 17.
We have performed an exhaustive computer search for 16-clue sudoku puzzles, and did not find any, thus proving that the answer is indeed 17.
In this article we describe our method and the actual search.
As a part of this project we developed a novel way for enumerating hitting sets.
The hitting set problem is computationally hard; it is one of Karp's 21 classic NP-complete problems.
A standard backtracking algorithm for finding hitting sets would not be fast enough to search for a 16-clue sudoku puzzle exhaustively, even at today's supercomputer speeds.
To make an exhaustive search possible, we designed an algorithm that allowed us to efficiently enumerate hitting sets of a suitable size.
\end{abstract}

\newpage
\tableofcontents

\section{Introduction}\label{sec:intro}

Sudoku is a logic puzzle --- one is presented with a \mbox{$9 \times 9$} grid, some of whose cells already contain a digit between 1 and 9; the task is then to complete the grid by filling in the remaining cells such that each row, each column, and each \mbox{$3 \times 3$} box contains the digits from 1 to 9 exactly once.
Moreover, it is always understood that any proper (valid) sudoku puzzle must have only one completion.
In other words, there is only a single solution, only one correct answer.
In this article, we consider the issue of \emph{how many} clues (digits) need to be provided to the puzzle solver in the beginning.
There are 81 cells in a grid, and in most newspapers and magazines, usually around 25 clues are given.
If too few clues are given at first, then there is more than one solution, i.e., the puzzle becomes invalid.
It is natural to ask how many clues are always needed.
This is the sudoku minimum number of clues problem:

\begin{center}
	\colorbox[rgb]{1,1,0.8}{
		\begin{minipage}[c]{9cm}
			\begin{center}
				What is the smallest number of clues that can be given such that a sudoku puzzle has a unique completion?
			\end{center}
		\end{minipage}}
\end{center}
More informally --- what is the smallest number of clues that you could possibly have?

There are puzzles known with only 17 clues, here is an example:
\begin{figure}[h]
	\begin{center}
	\includegraphics[scale=0.85]{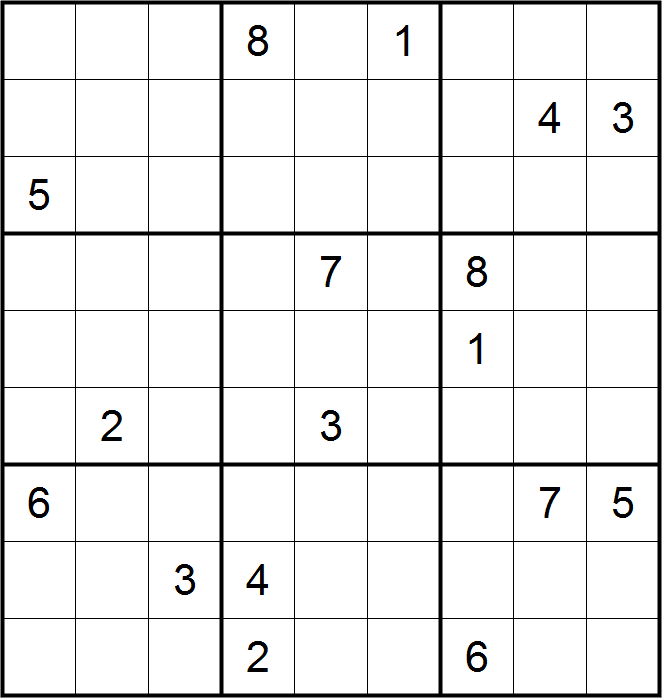}
	\caption{A 17-clue sudoku puzzle}
	\vspace{-15pt}
	\end{center}
\end{figure}

\noindent
However, nobody ever found any proper 16-clue puzzles.
Consequently, people started to conjecture that the answer to the minimum number of clues problem is 17.
We have proved this conjecture, and in this article we present our method.

The strategy we chose to solve this problem is a very obvious one: use a computer to exhaustively search through every possible sudoku solution grid, one by one, for a 16-clue puzzle.
So we took the viewpoint of considering each completed sudoku grid, one at a time, and looking for 16-clue puzzles whose solution was that grid --- the reader can think of these puzzles as being contained in the given grid.
Our year-long search, carried out at the Irish Centre for High-End Computing (ICHEC), turned up no proper 16-clue puzzles, but had one existed, we would have found it.

Due to the sheer number of sudoku solution grids a brute force search would have been infeasible, but we found a better approach to make this project possible.
Our software for exhaustively searching through a completed sudoku grid, named \emph{checker}, was originally released in 2006.
However, this first version was rather slow.
Indeed, the paper \cite{wu} estimates that our original \emph{checker} of late 2006 would take over 300,000 processor-years in order to search every sudoku grid.
With our new algorithm, this computation actually took only about 800 processor-years.
In the present article we describe both the original and the new version of \emph{checker}.
In particular, we describe our hitting set algorithm, which is at the heart of \emph{checker}.

The proof in this article is a computer-assisted proof of a mathematical theorem.
There are many precedents for this type of proof.
The most famous is without doubt the proof of the four-colour the\-orem by Appel and Haken \cite{appel}, see also \cite{robertson}.
Another more recent but also already legendary example of a computer-assisted proof is the proof of the Kepler conjecture by Thomas Hales \cite{hales}.
A very recent example is the proof of the weak Goldbach conjecture \cite{helfgott}.
There are a great many theorems published in the mathematical literature where the authors have used computers in their proof.
Comments about correctness always apply; we discuss this more with regard to our proof in Section~\ref{sec:correctness}.

As a side comment, there have been attempts to solve the sudoku minimum number of clues prob\-lem using mathematics only, i.e., without the help of computers.
However, nobody has made any real progress --- though it is easy to see why a 7-clue puzzle cannot have a unique solution,\footnote{In a 7-clue sudoku puzzle, the two missing digits can be interchanged in any solution to yield another solution.} a theoretical proof of the nonexistence of an 8-clue puzzle is still lacking.
This is far from the answer of 17, so that a purely mathematical solution to the minimum number of clues problem is a long way off.

\subsection{Summary description of method}\label{summmethod}

As just explained, the objective of this project was to prove that there are no 16-clue sudoku puzzles having exactly one completion, or to find such a puzzle, had one existed.
Our choice of method meant that there were three steps:
\begin{enumerate}
	\item make a catalogue of all completed sudoku grids --- obviously only finitely many grids exist;
	\item write \emph{checker}, i.e., implement efficient algorithms for exhaustively searching a given completed sudoku grid for puzzles having 16 clues and whose (unique) solution is the grid at hand;
	\item run through the catalogue of all completed grids and search every grid in turn using \emph{checker}.
\end{enumerate}
We will describe each of these subprojects in detail in Sections~\ref{sec:catalog}, \ref{sec:checker}--\ref{sec:newalgo} and \ref{sec:computation}, respectively. 

\subsection{Other applications of hitting set algorithms}\label{sec:applications}

Solving the sudoku minimum number of clues problem serves as a way to introduce our algorithm for enumerating hitting sets.
However, this algorithm has uses beyond combinatorics.
To begin with, it is in principle applicable to any instance of the hitting set problem, be it the \emph{decision version} (determin\-ing if hitting sets of a given size exist) or the \emph{optimization version} (finding a smallest hitting set).
Such situations occur in bioinformatics (gene expression analysis \cite{ruchkys}), software testing as well as computer networks \cite{abu}, and when finding optimal drug combinations \cite{vazquez,mellor}; the last of these papers lists protein network discovery, metabolic network analysis and gene ontology as further areas in which hitting set problems naturally arise.
Besides, our algorithm can be applied to the hypergraph transversal problem, which appears in artificial intelligence and database theory \cite{boros}.
Lastly, we note that the vertex cover problem from graph theory is a special case of the hitting set problem, and that the set cover problem is in fact equivalent to the hitting set problem.\footnote{We did not find any papers about the set cover problem that were useful to us in improving our hitting set algorithm.}
The former occurs in computational biology, see \cite{niederm}, while the latter finds application in reducing interference in cellular networks \cite{kuhn}.

\subsection{Miscellaneous comments}

This article is an update of a preprint, which we posted online \cite{checker} on $1^{st}$ January 2012 together with the full source code of the new \emph{checker}.
We shall not discuss any of the following topics here: how to solve or create sudoku puzzles, how to rate the difficulty of sudoku puzzles, or how to write a sudoku solver programme.
We emphasize that we are not saying that every completed sudoku grid contains a 17-clue puzzle (actually, very few do).
We \emph{are} saying that no grid contains any 16-clue puzzles.

\section{History}\label{sec:history}

The work on this project began in August 2005, when we started developing \emph{checker}, which to our knowledge was the first computer programme that made it possible to search a sudoku solution grid exhaustively for all \mbox{$n$-clue} puzzles, where both the grid and the number $n$ were supplied by the user.
The final release of this original version of \emph{checker} is of November~2006 \cite{checkerold}.
Throughout 2007, this project lay practically dormant, but in the second half of the following year we realized that there was considerable potential to make \emph{checker} faster.
Therefore, in early 2009 we got started implementing a new version of \emph{checker}, redesigned from the ground up and tailored to the case \mbox{$n=16$}, which we did not post online, however, until every sudoku grid had actually been exhaustively searched for 16-clue puzzles.
This wholly new \emph{checker} takes only a few seconds to search through a typical grid, whereas our original release from late 2006 needs about half an hour per grid on average.

To be able to do some early testing on their clusters, we were granted a Class C project account at the Irish Centre for High-End Computing (ICHEC) in July~2009.
In September of the same year we also successfully applied for a PRACE prototype award, which provided us with nearly four million core hours on JUGENE in J{\"u}lich, Germany --- then Europe's fastest supercomputer --- as well as time on a Bull cluster at CEA, France, a Cray XT5 cluster at CSC, Finland, and an IBM Power 6 cluster at SARA in the Netherlands.
The goal of this prototype project was to evaluate different loadbalancing strategies of the hitting set problem; the actual code we ran was an early version of the new \emph{checker}.

\subsection{Previous work by others}\label{subsec:previous}

\begin{itemize}
	\item In Japan, sudoku was introduced by the publisher Nikoli in the 1980s.
	Japanese puzzle creators have made puzzles with 17 clues, and with no doubt wondered whether 16 clues were possible. 
	Nikoli have a rule that none of their puzzles will have more than 32 clues.

	\item Over the years, people have collected close to 50,000 different 17-clue sudoku puzzles, which may be downloaded from Gordon Royle's homepage \cite{royle}.
	Most of these were found by Royle, who is a mathematician at the University of Western Australia and who compiled this list while searching for a 16-clue puzzle.
	On the sudoku forums \cite{forums} we became aware of a rather special grid (shown on p.~\pageref{img:sfgrid}), which has 29 different 17-clue puzzles, also discovered by Royle.
	Many considered this grid a likely candidate to contain a 16-clue puzzle.
	Actually, we initially started to work on \emph{checker} to search this particular grid.

	\item By the end of 2007, the sudoku minimum number of clues problem had been mentioned in sev\-eral journal publications \cite{delahaye,hayes,herzberg,pegg,taalman}.
	The first one is an article by the French computer scientist Jean-Paul~Delahaye entitled \emph{The Science behind Sudoku}, which was first published in the June 2006 issue of the \emph{Scientific American}.
	This article in fact quotes one of the authors of this study (Gary~McGuire) in conjunction with the minimum number of clues problem.

	\item Between autumn 2007 and spring 2008, a team at the University of Graz in Austria verified by distributed exhaustive computer search that there are no proper sudoku puzzles having eleven clues.
	They continued for about another year, and by May 2009, they had apparently completed most of the computations necessary to show that twelve clues are never sufficient for a unique solution either.
	Their stated aim was to build up to proving that no 16-clue sudoku puzzle exists, yet as of early 2010, that project seemed not to be active anymore \cite{graz}.

	\item Since 2007, Max~Neunh\"offer of the University of St.~Andrews in Scotland has lectured several times at different venues on the mathematics of sudoku, discussing in particular the minimum number of clues problem \cite{neunh}.
	According to the slides of his talks, Neunh\"offer appears to even have written a computer programme for searching a sudoku solution grid for 16-clue puzzles.

	\item In 2008, a 17-year-old girl submitted a proof of the nonexistence of a 16-clue sudoku puzzle as an entry to \emph{Jugend forscht} (the German national science competition for high-school students).
	She later published her work in the journal \emph{Junge Wissenschaft} (No.~84, pp.~24--31).
	However, when Sascha~Kurz, a mathematician at the University of Bayreuth, Germany, studied the proof closely, he found a gap that is probably very difficult, if not impossible, to fix.

	\item In mid-2010, the Bulgarian engineer Mladen Dobrichev released the initial version of his opensource tool \emph{GridChecker} \cite{dobrichev}; he has since provided several updates.
	Dobrichev's programme, also written in C++, basically does the same thing as our original \emph{checker}, although it is at least one order of magnitude faster.

	\item In October~2010, a group of computer scientists at the National Chiao Tung University, Taiwan, led by I-Chen Wu started a distributed search for 16-clue puzzles on the Internet using BOINC.	
	Their strategy is exactly the same as ours (i.e., consider each sudoku solution grid individually and exhaustively search for a 16-clue subset whose only completion is the given grid).
	This had become feasible since Wu's team had succeeded in speeding up our original \emph{checker} by a factor of 129, so that it would take them an estimated 2,400 processor-years to examine every grid for 16-clue puzzles.
	Wu also spoke about this in November~2010 at the International Conference on Technologies and Applications of Artificial Intelligence, and he and Hung-Hsuan Lin together published the article \emph{Solving the Minimum Sudoku Problem} in the conference proceedings, describing some of the techniques they had used to improve \emph{checker}.
	The complete details of their work are described in the article \cite{wu}.	
	As far as the progress of the BOINC search is concerned, according to that project's website, as of $31^{st}$ December 2011 --- the day before we announced our result ($\mathrm{arXiv\!:\!1201.0749v1}$) --- they had checked 1,453,000,000 sudokus (26.5 per cent of all cases that need to be considered, see Section~\ref{sec:catalog}).

	\item Toward the end of December 2011, the digital edition of the book \emph{Taking Sudoku Seriously: The Math Behind the World's Most Popular Pencil Puzzle} by Jason Rosenhouse and Laura Taalman of James Madison University in Harrisonburg, Virginia, came out; the printed version followed several weeks later \cite{rosenh}.
	This book devotes the whole of Section~9.4 \emph{The Rock Star Problem} to the minimum number of clues problem, saying:
	\begin{quote}	
	``If you can figure it out, you will be a rock star in the universe of people who care about such things. Granted, that is a far smaller universe than the one full of people who care about actual rock stars, but still, it would be great.''
	\end{quote}
\end{itemize}

Regarding the hitting set problem, we were surprised by the paucity of relevant literature, given the wide range of applications of an efficient algorithm for finding hitting sets, see Section~\ref{sec:applications}.
Most authors appear to have concentrated on the special case of the $d$-hitting set problem for small $d$, such as \mbox{$d = 3$}, where $d$ is the (maximum) number of elements in the sets to be hit.
However, to tackle the sudoku minimum number of clues problem efficiently, we would have needed a method for \mbox{$d = 12$} at least, so that these algorithms were not actually useful to us.
Some researchers have generalized their ideas to the case of arbitrary $d$, but the only such recent paper we could find describes an algorithm of complexity \mbox{$\mathrm{O}(\alpha^k + n)$}, where \mbox{$\alpha = d - 1 + \mathrm{O}(\frac{1}{d})$}, $k$ is the cardinality of the desired hitting set, and $n$ is the length of the encoding of the input \cite{niederm}.
In fact, this hitting set algorithm is somewhat similar to the one in our original version of \emph{checker}.
In a forthcoming article we shall therefore carry out the formal complexity analysis of our new algorithm.\footnote{``Often, performance measures the line between the feasible and the infeasible [...] If you're talking about doing stuff that nobody's done before, one of the reasons often that they haven't done it is, because it's too time-consuming, things don't scale, and so forth.  So that's one reason, is the feasible vs.\ infeasible.'' ---Charles E.~Leiserson, Professor of Computer Science, MIT, explaining why it is important to study algorithms and their performance (Source: MIT 6.046J / 18.410J Introduction to Algorithms, Fall~2005, Lecture~1, http://www.youtube.com/watch?v=JPyuH4qXLZ0 from 23:40 to 24:40)}
We estimate its average-case runtime complexity to be \mbox{$\mathrm{O}(d^{k-2})$} with any instance of the hitting set problem for which the sets that need to be hit are of comparable density as with the sudoku minimum number of clues problem.
 
\subsection{Heuristic arguments that 16-clue sudoku puzzles do not exist}\label{subsec:heuristics}

There are two heuristic arguments as to why 16-clue puzzles should not exist, which we present here.
The first argument is statistical, while the second one is a pattern that appears to correctly predict the minimum number of clues needed with general (\mbox{$n \times n$}) sudoku.

\subsubsection{A statistical observation}

For several years now, whenever people send Gordon Royle (see the second item in Section \ref{subsec:previous}) a list of 17-clue puzzles, there are usually not too many new puzzles.
One correspondent sent 700 puzzles, and it turned out that only 33 were new.
Assuming that both Royle and this correspondent had drawn their puzzles at random from the universe of all 17-clue puzzles, Ed Russell computed the maximum likelihood estimate for the size of the universe to be approximately 34,550 at a time at which Royle's list contained around 33,000 puzzles \cite{russell1}.
In retrospect this is an underestimate; nevertheless we may infer that the 49,151 puzzles on the most recent version of Royle's list must be almost all the 17-clue puzzles in existence.

On the other hand, among all sudoku solution grids known to contain at least one 17-clue puzzle, the highest number of such puzzles in any one grid is 29.\footnote{Exactly one grid having this many 17-clue puzzles is known; further there are four grids known containing 20, 14, 12 respectively 11 puzzles with 17 clues. All other grids arising from a puzzle on Royle's list have less than ten such puzzles.}
So presently there is no sudoku grid known with 30 or more 17-clue puzzles.
Given a (hypothetical) 16-clue sudoku puzzle, adding one clue to it in all possible ways results in 65 different 17-clue puzzles all having the same solution grid.
Because of the large gap between 29 and 65, assuming that Royle's list is practically complete, when we began running compute jobs for this project, it was already clear that a 16-clue puzzle was unlikely to exist.

One additional experimental finding that supports the assumption that Royle's list is nearly complete is the following.
When considering the actual solutions of all the puzzles on this list, then there are exactly 46,294 distinct sudoku grids (up to equivalence, see Section~\ref{sec:catalog}).
We took a random sample of 1 in 5,000 sudoku solution grids, exhaustively searched all grids in this sample for 17-clue puzzles, and got a hit in exactly nine cases, which is remarkable since \mbox{$46,\!294 / 5,\!000 \approx 9.26$}.

\subsubsection{A way of computing the minimum number of clues needed for general sudoku?}\label{subsubsec:formula}

In order to describe the second heuristic we first need to introduce some notation.
With \mbox{$n \times n$} sudoku, we denote the minimum number of clues needed for a puzzle to have a unique solution by \mbox{$\mathrm{scs}(n)$}, for \emph{smallest critical set}.
In general terms, critical sets are (minimal) subsets of a structure that determine the entire structure.
There have been papers on smallest critical sets in other combinatorial structures; for example, see \cite{bean} for \mbox{$8 \times 8$} Latin squares.
We will now demonstrate how to compute \mbox{$\mathrm{scs}(n)$} from only the total number of \mbox{$n \times n$} solution grids --- the results we get are correct for all $n$ for which the value of \mbox{$\mathrm{scs}(n)$} is currently known.
Here is a summary of the literature to date:
\smallskip
\begin{center}
	\begin{tabular}{|c|c|c|c|}
		\hline
		\textbf{grid dimension} ($n \times n$) & \textbf{boxes} & \textbf{total number of solution grids} & $\mathbf{scs}(n)$ \\\hline
		$4 \times 4$ & $2 \times 2$ & $288$ & $4$ \\\hline
		$6 \times 6$ & $2 \times 3$ & $28,\!200,\!960$ & $8$ \\\hline
		$8 \times 8$ & $2 \times 4$ & $29,\!136,\!487,\!207,\!403,\!520$ & $14$ \\\hline
		$9 \times 9$ & $3 \times 3$ & $6,\!670,\!903,\!752,\!021,\!072,\!936,\!960$ & $17$ \\\hline
	\end{tabular}
\end{center}
\smallskip
For the total number of grids, \cite{wiki} has references.
A proof that \mbox{$\mathrm{scs}(4) = 4$} may be found, e.g., in \cite{taalman}.
The result that \mbox{$\mathrm{scs}(6) = 8$} is due to Ed~Russell, who, back in 2006, had checked that no proper 7-clue puzzle exists for \mbox{$6 \times 6$} sudoku \cite{russell2}.
In contrast, the \mbox{$8 \times 8$} case remained open until November~2011, when Christoph Lass of the University of Greifswald in Germany announced that he had proved that \mbox{$\mathrm{scs}(8) = 14$} \cite{lass1}.
Although Lass, too, performed a computer search, his approach was not to inspect every completed \mbox{$8 \times 8$} sudoku grid, but rather to generate all 13-clue puzzles (up to equivalence) and then see if any of them had a unique completion.
Observe now that, for these four cases,
\begin{eqnarray*}
	4 \times 6 \times 8 \quad < & 288 & < \quad 4 \times 6 \times 8 \times 10, \\
	6 \times 8 \times \cdots \times 18 \quad < & 2.82 \times 10^7 & < \quad 6 \times 8 \times \cdots \times 18 \times 20, \\
	8 \times 10 \times \cdots \times 32 \quad < & 2.91 \times 10^{16} & < \quad 8 \times 10 \times \cdots \times 32 \times 34, \\
	9 \times 11 \times \cdots \times 39 \quad < & 6.67 \times 10^{21} & < \quad 9 \times 11 \times \cdots \times 39 \times 41.
\end{eqnarray*}
So in each case, if $T$ is the total number of completed \mbox{$n \times n$} sudokus and \mbox{$C := \mathrm{scs}(n)$}, then
\begin{equation*}
	\underbrace{n \times (n + 2) \times \cdots \times (n + 2(C-2))}_{(C-1)\;\textrm{factors}} \;\; < \;\; T \;\; < \;\; \underbrace{n \times (n + 2) \times \cdots \times (n + 2(C-1))}_{C\;\textrm{factors}}.
\end{equation*}
Moreover, for \mbox{$16 \times 16$} sudoku, according to \cite{lass2}, the conjectured size of a smallest critical set is 56, while from \cite{wiki} the estimate of the total number of completed \mbox{$16 \times 16$} grids is about \mbox{$5.96 \times 10^{98}$}.
If we assume that these two values are correct, then both inequalities hold.
Note also that the first of the above inequalities is logarithmically quite sharp in all five cases just considered.
\section{The catalogue of all sudoku grids}\label{sec:catalog}

This section describes Step~1 of Section~\ref{summmethod}.
As already stated in the last section, in total there are
	$$6,\!670,\!903,\!752,\!021,\!072,\!936,\!960 \;\approx\; 6.7 \times 10^{21}$$
completed sudoku grids \cite{felgen}.
However, with this project, it was not necessary to analyze all of them.
For instance, even though relabelling the digits of a given grid yields a different grid, clearly this will not actually change the substance of the grid because the individual digits carry no significance --- any nine different symbols could be used.
There are several other such \emph{equivalence transformations}, e.g., rotating the grid or taking its transpose, or interchanging the first with the second row, none of which alter the property of containing a 16-clue puzzle.
Technically, we introduce an equivalence relation on the set of all completed sudoku grids, where two grids are equivalent if one may be obtained from the other by applying one or more of the equivalence transformations.
Whether a 16-clue puzzle exists in a particular grid is then an invariant of its equivalence class.
Hence we only needed to search one rep\-resentative from each grid equivalence class for 16-clue puzzles in order to solve the sudoku minimum number of clues problem.

\subsection{Equivalence transformations}\label{sec:EqRel}
In order to give a concise description of the equivalence transformations, we introduce the following terminology.
A \emph{band} in a sudoku grid is either the set of rows~1--3, rows~4--6, or rows~7--9; similarly a \emph{stack} denotes either the set of columns 1--3, columns 4--6, or columns 7--9.
(Just like with matrices in linear algebra, the rows and columns of a sudoku grid are numbered from top to bottom respectively left to right.)
This illustration shows the three bands and stacks in an empty grid:
\smallskip
\begin{figure}[h]
	\begin{center}
	\begin{tabular}{lr}
		\begin{minipage}{6cm}
			\begin{center}
				\includegraphics[scale=0.5]{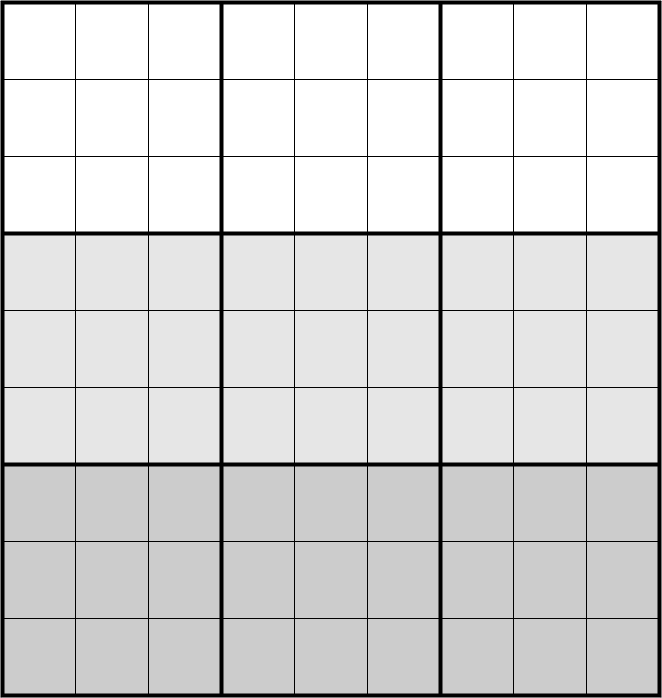}
			\end{center}
		\end{minipage}
	&
		\begin{minipage}{6cm}
			\begin{center}
				\includegraphics[scale=0.5]{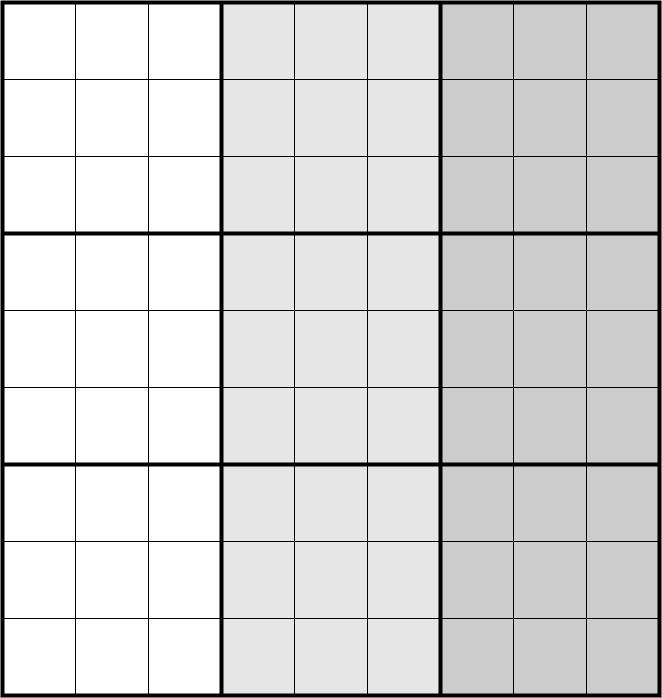}
			\end{center}
		\end{minipage}
	\end{tabular}
	\caption{The bands and stacks in an empty sudoku grid}
	\vspace{-15pt}
	\end{center}
\end{figure}

\smallskip\noindent
Here are the actual equivalence transformations, which other authors sometimes refer to as \emph{symmetry operations} --- note that reflections and rotations are already included in these:
\begin{enumerate}
\item relabelling the digits 1--9, i.e., replacing each instance of the digit $d$ by $\pi(d)$ for some \mbox{$\pi \in S_9$}, for all \mbox{$d=1, \ldots,9$};
\item permuting the rows (columns), by
  \begin{enumerate}
 	\item swapping bands (stacks), or
	\item swapping the three rows (columns) within a given band (stack);
	\end{enumerate}
\item transposing the grid.
\end{enumerate}
We call two sudoku solution grids \emph{equivalent} if one may be obtained from the other by any sequence of the equivalence transformations listed above.
It is routine to check that this definition gives rise to an equivalence relation in the mathematical sense on the set of all sudoku solution grids.
A sequence of these transformations mapping a grid to itself is called an \emph{automorphism} of the grid.
The possible automorphism groups have been determined \cite{russell3}, although the vast majority of grids do not actually possess any nontrivial automorphisms.
The largest automorphism group occurring has order 648, and is realized, up to equivalence, by exactly one grid (see Figure~\ref{fig:canonical}).

Now note that the operation of relabelling the digits commutes with all other equivalence transformations.
Also, every row operation commutes with each of the column operations.
Finally, swapping the $i^{th}$ with the $j^{th}$ row (column), \mbox{$1 \leq i < j \leq 9$}, and then transposing the resulting grid is clearly no different from first taking the transpose and then interchanging the $i^{th}$ with the $j^{th}$ column (row).
The consequence of all this is that if $G$ and $H$ are equivalent grids, then there exists a digit permutation $\pi$, as well as a row permutation $\rho$ and a column permutation $\sigma$ such that either
\begin{equation}\label{eq:equivalence}
	H \; = \; (\pi \circ \sigma \circ \rho)(G) \qquad\textrm{ or }\qquad H \; = \;(\pi \circ \sigma \circ \rho)(G^T).
\end{equation}
The following easy result shows that the property of containing a 16-clue puzzle is an invariant of the equivalence class of a grid.

\begin{lemma}
  Suppose that a sudoku solution grid $G$ has a 16-clue puzzle.
  Then all grids equivalent to $G$ contain a 16-clue puzzle.
\end{lemma}
\begin{proof}
	Let \mbox{$P \subset G$} be a 16-clue puzzle.
	Then $P^T$ is a 16-clue puzzle contained in $G^T$, because if $H$ is any sudoku solution grid containing $P^T$, then \mbox{$P = (P^T)^T \subset H^T$}, but as $G$ is the unique completion of $P$ it follows that \mbox{$H^T = G$}, i.e., \mbox{$H = G^T$}, so that $G^T$ is in fact the only completion of $P^T$.
	
	Next let $\rho$ be any permutation of the rows of a grid.
	We claim that $\rho(P)$ is a proper 16-clue puzzle of $\rho(G)$.
	For if $H$ is any completion of $\rho(P)$, then \mbox{$P = \rho^{-1}(\rho(P)) \subset \rho^{-1}(H)$}, and by uniqueness of $G$ it  follows that \mbox{$\rho^{-1}(H) = G$}, i.e., \mbox{$H = \rho(G)$}.
	Hence \mbox{$\rho(G)$} is the only completion of \mbox{$\rho(P)$}.
	
	A similar argument shows that if $\sigma$ is any rearrangement of the columns and $\pi$ is any permutation of the digits 1--9, then $\sigma(P)$ and $\pi(P)$ are 16-clue puzzles of the grids \mbox{$\sigma(G)$} and \mbox{$\pi(G)$}, respectively.
	On combining everything we just noted, the statement of the lemma now follows directly from (\ref{eq:equivalence}).
\end{proof}
An obvious question to ask is: how many sudoku solution grids are there, up to equivalence?
This is a natural mathematical question in its own right, but it was especially relevant for this project since we only needed to search one representative from each grid equivalence class for 16-clue puzzles, by the preceding lemma.

\subsection{Using the sudoku group and Burnside's lemma to count equivalence classes}

The set of all equivalence transformations defined above forms a group with respect to composition~of functions, called the \emph{sudoku group}, which acts on the set of all sudoku solution grids.
The orbits under this group action are then precisely the grid equivalence classes.
Of course the reason for employing group theory is that this makes it easier to determine the total number of equivalence classes.
So if
	$$T \;:=\; \{\sigma \in S_9 \mid \forall\; 1 \leq i < j \leq 9 \,:\, \lceil i / 3 \rceil = \lceil j / 3 \rceil \rightarrow \lceil \sigma(i) / 3 \rceil = \lceil \sigma(j) / 3 \rceil \},$$
then (as a \emph{set} only --- see (\ref{eq:sudokugroup}) below) the sudoku group is \mbox{$S_9 \times T \times T \times C_2$}.
In this cartesian product, $S_9$ corresponds to the digit renumberings, the two copies of $T$ provide the row and column permuta\-tions, while $C_2$ (the cyclic group having two elements) is appended to specify if the transpose is being taken.
Note that, in the definition of $T$, the condition \mbox{$\lceil i / 3 \rceil = \lceil j / 3 \rceil$} only if \mbox{$\lceil \sigma(i) / 3 \rceil = \lceil \sigma(j) / 3 \rceil$}, for all \mbox{$1 \leq i < j \leq 9$}, ensures that every \mbox{$\sigma \in T$} preserves bands (stacks) as the index of the band (stack) containing the $i^{th}$ row (column) is given by \mbox{$\lceil i/3 \rceil$}.
Moreover, the number of elements in the group $T$ equals \mbox{$9 \times 2 \times 6 \times 2 \times 3 \times 2 = 1,\!296$} ---
if \mbox{$\sigma \in T$}, then there are nine possibilities for \mbox{$\sigma(1)$}, but once the value of \mbox{$\sigma(1)$} has been chosen, only two possibilities remain for \mbox{$\sigma(2)$} because of the requirement that \mbox{$\lceil \sigma(1) / 3 \rceil = \lceil \sigma(2) / 3 \rceil$}, and since we further need that \mbox{$\lceil \sigma(2) / 3 \rceil = \lceil \sigma(3) / 3 \rceil$}, the choice of \mbox{$\sigma(2)$} already implies the value of \mbox{$\sigma(3)$}; similarly there are six possibilities for \mbox{$\sigma(4)$}, and with \mbox{$\sigma(4)$} chosen, just two possibilities for \mbox{$\sigma(5)$} remain etc.
So \mbox{$\#T = 1,\!296 = 6^4$}, which is in line with what one would expect as there are \mbox{$3! = 6$} ways to arrange the bands, as well as $3!$ permutations of the rows in each of the three bands.
Actually, an isomorphism between \mbox{$S_3 \times S_3 \times S_3 \times S_3$} and $T$ is given by the map
	$$(\sigma_1, \sigma_2, \sigma_3, \gamma) \;\mapsto\; \bigl(\{1,\ldots,9\} \ni i \mapsto \sigma_b(i - 3(\lceil i/3 \rceil - 1)) + 3(b-1)\bigr),\quad \sigma_1,\sigma_2,\sigma_3,\gamma \in S_3,$$
where \mbox{$b := \gamma(\lceil i/3 \rceil)$} and the group multiplication in \mbox{$S_3 \times S_3 \times S_3 \times S_3$} is defined by
	$$(\sigma_1, \sigma_2, \sigma_3, \gamma) \ast (\tau_1, \tau_2, \tau_3, \lambda) \;=\; (\sigma_1 \circ \tau_{\gamma^{-1}(1)}, \sigma_2 \circ \tau_{\gamma^{-1}(2)}, \sigma_3 \circ \tau_{\gamma^{-1}(3)}, \gamma \circ \lambda),$$
for \mbox{$\sigma_i, \tau_i, \gamma, \lambda \in S_3, i = 1, 2, 3$}.
The reader may have noticed that \mbox{$S_3 \times S_3 \times S_3 \times S_3$} equipped with this particular composition law is in fact \mbox{$S_3 \;\mathrm{wr}\; S_3$}, the \emph{wreath product} of $S_3$ by itself such that $S_3$ acts on \mbox{$\{1,2,3\}$} in the natural fashion.
Even so, we shall briefly recall how to construct wreath products.

For two groups $N$ and $H$ and a homomorphism \mbox{$\psi\colon H \to \mathrm{Aut}(N)$}, the (external) \emph{semidirect product} of $N$ by $H$ with respect to $\psi$, denoted \mbox{$N \rtimes_\psi H$}, is defined like so.
As a set, \mbox{$N \rtimes_\psi H := N \times H$}, the ordinary cartesian product of $N$ and $H$; the group multiplication is given by
$$\mbox{$(n_1, h_1) \ast (n_2, h_2) \;=\; (n_1\psi(h_1)(n_2),h_1 h_2)$},$$
where \mbox{$n_1, n_2 \in N$} and \mbox{$h_1, h_2 \in H$} (a straightforward calculation confirms that all group axioms are satisfied).
The wreath product is then a special case of the semidirect product.
Suppose that $G$ and $H$ are groups such that $H$ acts on the set \mbox{$\{1,\ldots,n\}$} for some positive integer $n$, and let \mbox{$\phi\colon H \to S_n$} be the homomorphism corresponding to this (left) action.
Set \mbox{$N = G^n$}, the direct product of $n$ copies of $G$.
For each \mbox{$h \in H$}, the map \mbox{$\pi_h\colon N \to N, (g_1, \ldots, g_n) \mapsto (g_{\phi(h)(1)},\ldots,g_{\phi(h)(n)}),$} is clearly a group isomorphism, i.e., \mbox{$\pi_h \in \mathrm{Aut}(N)$}, for all \mbox{$h \in H$}.
So there is a mapping \mbox{$\pi\colon H\to \mathrm{Aut}(N), h \mapsto \pi_{h^{-1}},$} which is in fact a homomorphism (routine verfication).
The wreath product of $G$ by $H$ with respect to the associated group action, written \mbox{$G \;\mathrm{wr}\; H$}, is finally defined to be \mbox{$N \rtimes_\pi H$}, the semidirect product of \mbox{$N = G^n$} by $H$ using the homomorphism \mbox{$H \to \mathrm{Aut}(N)$} induced by the action of $H$ on \mbox{$\{1,\ldots,n\}$}.
Explicitly, the group multiplication in \mbox{$G \;\mathrm{wr}\; H$} is given by (\mbox{$g_1,\ldots,g_n,g_1',\ldots,g_n' \in G, h,h' \in H$})
$$(g_1, \ldots, g_n, h) \ast (g_1', \ldots, g_n', h') \;=\; (g_1g_{\phi(h^{-1})(1)}',\ldots,g_ng_{\phi(h^{-1})(n)}',hh').$$
After this digression, we now continue with the actual definition of the sudoku group and its action on the set of all solution grids.
In the beginning of this subsection, we already defined the sudoku group to be \mbox{$S_9 \times T \times T \times C_2$} as a set only, the reason being that the law of composition is again not simply that of the direct product, which is so because the various types of equivalence transformations do not all commute with each other.

Specifically, as already noted in the discussion right before (\ref{eq:equivalence}), the relabelling of digits commutes with all other equivalence transformations, and every row permutation commutes with every column permutation.
Thus the group generated by just these three types of transformations actually is a direct product, namely \mbox{$S_9 \times T \times T$}.
However, first permuting the rows and then taking a grid's transpose is not the same as performing these operations in reverse order; on taking the transpose first, we need to apply the respective permutation to the columns instead of the rows.
So this leads very naturally to a semidirect product because the semidirect product, really being a generalization of the direct product, provides exactly the additional flexibility needed to handle this situation.
In fact, the homomorphism \mbox{$\psi\colon C_2 \to \mathrm{Aut}(S_9 \times T \times T)$} that yields the correct semidirect product in our case is very simple --- if the group $C_2$ is generated by $g$, i.e., if $g$ is the nonidentity element of $C_2$, then under $\psi$,
$$g \;\mapsto\; \bigl( S_9 \times T \times T \ni (\pi,\rho,\sigma) \mapsto (\pi,\sigma,\rho)\bigr).$$
So $\psi(g)$ merely swaps the row and column permutations (the second respectively third component of a triple \mbox{$(\pi,\rho,\sigma) \in S_9 \times T \times T$}).
At last we can properly and fully define the sudoku group; it is
\begin{equation}\label{eq:sudokugroup}
	(S_9 \times T \times T) \rtimes_\psi C_2.
\end{equation}
This group has order \mbox{$9! \times 6^4 \times 6^4 \times 2$}.
The permutation part only (i.e., when the digit relabellings are omitted) is a subgroup of order \mbox{$6^4 \times 6^4 \times 2 = 3,\!359,\!232$}.
For defining the group action on the set of all completed sudoku grids, we identify a grid with the corresponding element of \mbox{$\mathbb{M}_{9 \times 9}(\{1,\ldots,9\})$}, i.e., we think of a grid as a \mbox{$9 \times 9$} matrix \mbox{$[a_{i,j}]$} with integral coefficients between 1 and 9.
An arbitrary member \mbox{$(\pi,\rho,\sigma,\ell)$} of the sudoku group \mbox{$(S_9 \times T \times T) \rtimes_\psi C_2$} then acts on a grid \mbox{$[a_{i,j}]$} as follows:
$$(\pi,\rho,\sigma,\ell).[a_{i,j}] \;=\; \left\{\begin{array}{cl}
	[\pi(a_{\rho^{-1}(i),\sigma^{-1}(j)})], & \ell = 1_{C_2}, \\
		
	[\pi(a_{\sigma^{-1}(j),\rho^{-1}(i)})], & \textrm{otherwise}.
\end{array}\right.$$
The check that this really defines a group action is again a completely routine calculation.
(Here, \mbox{$1_{C_2}$} denotes the identity element of the group $C_2$.)
Determining the precise number of orbits of this action can be done by appealing to Burnside's lemma, sometimes called Burnside's counting theorem or the Cauchy-Frobenius lemma.
This was carried out by Ed~Russell and Frazer~Jarvis in 2006 \cite{russell}.
Using a very clever computer calculation --- taking only one second --- they proved that there are exactly
$$5,\!472,\!730,\!538 \;\approx\; 5.5 \times 10^{9}$$
orbits.
So this is the number of equivalence classes, i.e., there are \mbox{$5,\!472,\!730,\!538$} \emph{essentially different} sudoku grids in total.
Later this result was independently established by Kjell Fredrik Pettersen.

We remark that this provides another way to see that most sudoku grids do not have any nontrivial automorphisms, as noted earlier.
As if we divide the total number of grids (see p.~\pageref{sec:catalog}) by the number of essentially different grids, we get that the average number of grids per equivalence class is
	$$\frac{6,\!670,\!903,\!752,\!021,\!072,\!936,\!960}{5,\!472,\!730,\!538} \;\approx\; 9! \times 3,\!359,\!058.6,$$
which is roughly \mbox{$99.9948\%$} of the order of the sudoku group --- in actual fact there are only \mbox{$560,\!151$} essentially different grids having nontrivial automorphism group \cite{gsf1}.

\subsection{Enumerating representatives}

For this project it was not enough to just know the number of sudoku grid equivalence classes; it was necessary to enumerate a complete set of representatives, and store these in a file.
Fortunately for us, Glenn Fowler had already written a computer programme named \emph{sudoku} \cite{gsf2} that is capable, among many other things, of enumerating all inequivalent completed sudoku grids in their respective minlex representation --- the \emph{minlex representation}, or \emph{minlex form}, of a grid is the member of its equivalence class that comes first when all grids in that equivalence class are ordered lexicographically.

Uncompressed, the \emph{catalogue} (of grids) would require about 418~GB of storage space, but Fowler also developed a data compression algorithm to store the catalogue in under 6~GB.
Note the genuinely amazing compression rate; each grid occupies on average only a little more than one byte.
Fowler has kindly shared his executables, and we have used them in order to generate and compress the complete list of essentially different sudoku solution grids in minlex form.
We were thus able to store the entire catalogue on a single DVD.

We verified the completeness of the catalogue as follows.
We wrote a small tool that computes the minlex form of a given sudoku solution grid, and then used this tool to double-check that each grid in the catalogue actually was in minlex form already.
Further, we made sure that the catalogue itself was ordered lexicographically and that there were no repetitions.
In this way we knew that no two grids in the catalogue were members of the same equivalence class.
Because of the result on the total number of equivalence classes from the last subsection, that was all it took to be certain that the catalogue was really complete.
\section{Overall strategy of checker}\label{sec:checker}

This section concerns Step~2 of Section~\ref{summmethod}.
The obvious algorithm for exhaustively searching a given completed sudoku grid for 16-clue puzzles is to simply consider all subsets of size~16 of that grid and test if any have a unique completion.
Even though such a \emph{brute force} algorithm is effective in that the minimum number of clues problem could in principle be solved that way, the actual computing power required would be far too great.
In fact, searching just one grid in this way would be infeasible as
	$$\binom{81}{16} \;\approx\; 3.4 \times 10^{16}.$$
Fortunately, with just a little theory, the number of possibilities to check can be reduced quite dramat\-ically.
Very briefly, it is possible to identify regions (subsets) in a given completed sudoku grid, called \emph{unavoidable sets}, that must always have at least one clue from any proper puzzle in that grid; once all the unavoidable sets have been identified for the given grid, the problem of finding 16-clue puzzles in that grid can be solved by finding \emph{hitting sets} --- sets of size~16 that intersect (hit) each member of the collection of unavoidable sets.
This is a much smaller problem, i.e., the approach just outlined greatly reduces the compute time taken per grid compared to the brute force way.

In practice, we use only a restricted collection of unavoidable sets, namely those having no more than twelve elements.
Obviously any proper 16-clue puzzle necessarily hits in particular the unavoidable sets in this subcollection, but not conversely --- a hitting set for this smaller collection may miss other unavoidable sets, of size~13 or greater.
So the hitting sets for this collection are, a priori, merely \emph{trial} (or \emph{candidate}) sudoku puzzles, in the sense that a hitting set will not usually be a proper 16-clue puzzle but every proper 16-clue puzzle will be found in this way.
For this reason we need to perform the following extra step: each hitting set enumerated, i.e., each trial 16-clue sudoku puzzle produced, is run through a sudoku solver procedure, to see if there are multiple solutions; this effectively checks if there are any unavoidable sets in the grid having empty intersection with the hitting set in question.
Thus the overall strategy we use in our programme \emph{checker} to exhaustively search for 16-clue puzzles in a sudoku solution grid may be summarized as follows:
\begin{enumerate}
	\item[(C1)] find a sufficiently powerful collection of unavoidable sets for the given grid;
	\item[(C2)] enumerate all hitting sets of size~16 for this collection, i.e., find every possible combination of 16~clues that intersect all the unavoidable sets found in Step~C1;
	\item[(C3)] check if any of the hitting sets found in the previous step is a proper 16-clue puzzle, i.e., test if any of these hitting sets uniquely determine the given grid, by running each hitting set through a sudoku solver procedure.
\end{enumerate}
Importantly, these three steps do not equally contribute to the running time of \emph{checker}.
It is Step~C2, the enumeration of hitting sets, that typically consumes the bulk (around 95 per cent) of CPU cycles, which however should come as no surprise as the hitting set problem is NP-complete \cite{karp}.
Of course, being NP-complete is an asymptotic statement, so it does not really say anything about the difficulty of a problem of a fixed size.
On the other hand, experience suggests that \mbox{$k=16$} is often large enough a hitting set size for a standard backtracking algorithm to be too inefficient to solve billions of instances of the hitting set problem, as it indeed turned out.
We will now provide an introductory description of Steps~C1, C2 and C3; all details are given in Sections~\ref{sec:unavsets}--\ref{sec:newalgo}.

\subsection{Unavoidable sets}\label{subsec:unavsets}

This subsection provides a sketch of Step~C1.
The idea of an \emph{unavoidable set} in a completed sudoku grid is easily explained by example.\footnote{Later we learned that the analogous concept had been defined previously in the study of Latin squares, where what we refer to as unavoidable sets are called \emph{Latin trades}.}
So let us consider the following grid:
\begin{figure}[h]
	\begin{center}
	\includegraphics[scale=0.85]{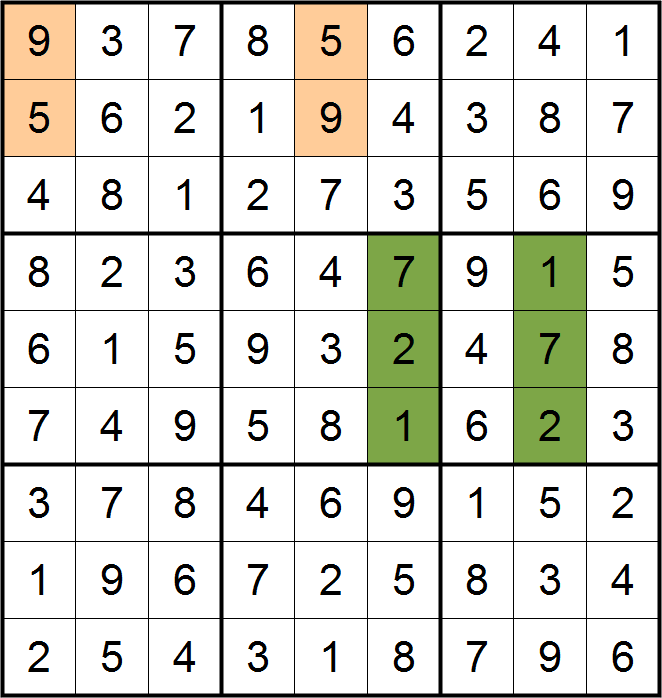}
	\caption{Two unavoidable sets in a sudoku solution grid}
	\vspace{-15pt}
	\end{center}
\end{figure}

\noindent
The reader can see that if the digits 5 and 9 are interchanged \emph{among the four orange cells only} (rows~1 and 2, columns~1 and 5), then a different correctly completed grid is obtained.
In other words, deleting these four digits results in a 77-clue sudoku puzzle with two completions.
Accordingly, in any puzzle with the above grid as the only possible answer, one of the four orange cells must contain a clue since a puzzle not having any of these four as a clue would have at least two completions and therefore not be a valid puzzle.
We say that the set of the orange cells is \emph{unavoidable} --- we cannot avoid having a clue from it.
Another example may be found in columns~6 and 8.
Replacing the contents of the three green cells in column 6 with those of the three green cells in column 8, i.e., swapping the digits in the cells in the middle band in column 6 with those in column 8, again yields a different sudoku solution grid.
So the six green cells (rows 4--6, columns 6 and 8) also form an unavoidable set.

In general, a subset of a sudoku solution grid is called unavoidable if it is possible to permute the contents of its cells, leaving the other cells unchanged, such that a different grid results.
(We will give a more formal definition in Section~\ref{sec:unavsets}.)
Directly from this definition, if a set of clues in a grid does not intersect every unavoidable set, then it cannot be used as the set of starting clues for a sudoku puzzle since there will be more than one solution.
Equivalently, any set of clues for a proper puzzle must use at least one clue from every unavoidable set.
By the way, the proof that seven clues are never enough for a unique completion (given in Section~\ref{sec:intro}) can be seen as follows: the \mbox{$\binom{9}{2}$} subsets consisting of the eighteen occurrences of two digits in a grid are all unavoidable, and seven clues must always miss one of these sets.

\subsubsection*{Finding unavoidable sets}

We find unavoidable sets in a given completed sudoku grid using a pattern-matching algorithm --- in \emph{checker}, there is a list of essentially different (with respect to the equivalence transformations defined in Section~\ref{sec:catalog}) types of unavoidable sets of size~12 or below; the actual search for unavoidable sets in a sudoku grid goes by effectively generating all grids equivalent to the given grid, and then comparing the list of essentially different types of unavoidable sets with each grid generated.
In case of a match, the corresponding unavoidable set in the original (given) sudoku solution grid is obtained by applying the inverse of the transformation used to generate the grid that gave the match.
Our \emph{checker} of 2006 took about 30~seconds per grid to find these unavoidable sets.
However, back then a lot of effort was duplicated by the procedure $\mathrm{FindUnavoidableSets}$.
The logic in the new \emph{checker} eliminates most of this redundancy (see Section~\ref{subsec:blueprints}), and takes less than one twentieth of a second for the same task.

\subsubsection*{How the unavoidable sets influence the running time of checker}

As mentioned earlier, when searching a sudoku solution grid for 16-clue puzzles using \emph{checker}, most of the compute time is usually spent on enumerating hitting sets for the collection of unavoidable sets.
So the family of unavoidable sets --- the \emph{structure} of the grid --- is what determines how computationally intensive each grid is.
A typical grid has around 360 unavoidable sets with up to twelve elements.
However, there is considerable variability in the structure between grids --- some grids have as many as 500 unavoidable sets (of size~12 or below), while others have fewer than 200 --- and thus the time it takes to search for a 16-clue puzzle also varies considerably among grids.
One might think that those grids having the most unavoidable sets are the quickest to search through, but that is not actually the case.
Interestingly, it turns out that grids with many unavoidable sets often have relatively few small unavoidable sets (of either size~4 or size~6), and the latter are an important factor in the running time of \emph{checker}.

\subsection{Hitting sets}\label{subsec:sketchC2}

In this subsection we provide a sketch of Step~C2, i.e., we will outline how we enumerate hitting sets.
Assume therefore that we have already found a sufficiently powerful family of unavoidable sets for a given sudoku grid.
As explained in the beginning of this section, if a 16-clue puzzle exists within this grid, then it intersects (hits) each of the unavoidable sets.
So we are faced with a common situation in computational discrete mathematics --- we need to solve an instance of the \emph{hitting set problem}.
Recall that the hitting set problem is formally defined as follows: given a finite set $U$ (the \emph{universe}), a family \mbox{$\mathcal{F}$} of subsets of $U$, and a positive integer $k$, the task is to find all hitting sets for $\mathcal{F}$ having $k$ elements, where a hitting set for $\mathcal{F}$ is simply a subset of $U$ that intersects every one of the sets in $\mathcal{F}$.
In our case, the universe $U$ is the given sudoku solution grid, the family $\mathcal{F}$ of subsets of $U$ consists of some of the unavoidable sets for this grid, and \mbox{$k = 16$} as we are interested in hitting sets having 16~elements.
We reiterate that for \emph{checker} to work as expected, i.e., for the search to be exhaustive, it is not essential to use \emph{all} unavoidable sets; any subcollection, no matter how small, of the collection of all unavoidable sets will produce correct results, by Step~C3.
For the actual computations we of course worked with a collection that minimizes the running time of \emph{checker} --- it turned out that the family of unavoidable sets with at most twelve elements gives the overall best results.

So we needed an efficient hitting set algorithm for this project.
In the original version of \emph{checker} we used the obvious backtracking algorithm, modified such that every hitting set is enumerated once only; full details of that are given in Section~\ref{sec:orgalgo}.
The hitting set algorithm in the new \emph{checker} has three main improvements over the original algorithm, which are explained in Section~\ref{sec:newalgo}.
The most important of these three improvements will be introduced now.
It is a technique to prune the search tree so as to abandon the traversal of a particular branch earlier than usual --- with a standard hitting set algorithm, detection of failure, and hence backtracking, does not begin until after the last (\mbox{$k^{th}$}) element has been added to the hitting set under construction.

\subsubsection*{Pruning the search tree using higher-degree unavoidable sets}

There is a natural generalization of the concept of an unavoidable set in a sudoku solution grid, which we call higher-degree unavoidable set --- a \emph{higher-degree} unavoidable set is a subset of the given grid from which more than one clue (say $m$ clues) is needed for a uniquely completable puzzle; the \emph{degree} of such a set is $m$.
For example, an unavoidable set of degree~3 is an unavoidable set such that every puzzle in the grid in question must have at least three elements from that set in order to have a unique solution.
So suppose that, when enumerating hitting sets of size~\mbox{$k > 1$} for a family $\mathcal{F}$ of unavoidable sets in a sudoku grid $G$, we are further given families \mbox{$\mathcal{F}^{(2)}, \ldots, \mathcal{F}^{(k)} \subset 2^G$} such that each member of \mbox{$\mathcal{F}^{(d)}$} is unavoidable of degree~$d$, \mbox{$d = 2, \ldots, k$}.
We will now show how these families of higher-degree unavoidable sets provide a very simple way to facilitate earlier backtracking, thereby speeding up the algorithm quite significantly.
Recall that, in general, backtracking begins as soon as it is clear that the solution being constructed (partial solution) cannot possibly be completed to an actual solution to the problem in question.

Assume that \mbox{$k-1$} clues have been selected already, i.e., suppose that \mbox{$k-1$} elements have so far been added to the hitting set under construction.
If there exists a member of \mbox{$\mathcal{F}^{(2)}$} (an unavoidable set of degree~2) not containing any of the \mbox{$k-1$} clues of our partial hitting set, then we may immediately stop and backtrack --- we already know that this \mbox{$(k-1)$}-clue set can never be completed to a $k$-clue hitting set, because any unavoidable set of degree~2 requires by definition at least two clues but there is only one more clue left to draw.
(In other words, no matter how the \mbox{$k^{th}$} clue is chosen, the resulting puzzle will have multiple completions.)
Similarly, if after selecting \mbox{$k-2$} clues there is a set in \mbox{$\mathcal{F}^{(3)}$} that does not contain any of the clues picked so far, then we can again backtrack immediately as each element of \mbox{$\mathcal{F}^{(3)}$} is an unavoidable set of degree~3 and therefore requires at least three clues; with only two clues to go, there is no point continuing.
More generally, if \mbox{$k-d+1$} clues have been drawn for some $d$, \mbox{$1 < d \leq k$}, and there is a set in \mbox{$\mathcal{F}^{(d)}$} that does not contain any of these clues, i.e., if there is an unavoidable set of degree~$d$ not yet hit, then we do not need to branch any further but may instead backtrack at once.

How do we actually obtain higher-degree unavoidable sets?
By taking disjoint unions of ordinary unavoidable sets.
As an example, consider once again the grid shown at the beginning of Section~\ref{subsec:unavsets}.
In this grid, two unavoidable sets are highlighted, one having four elements (the orange cells), and the other having six elements (green cells).
Note that these two unavoidable sets are disjoint, so that it is impossible to hit both with just a single clue.
Therefore at least two clues are required, i.e., the union of these two sets is an unavoidable set of degree~2.
We call a pair of disjoint unavoidable sets a \emph{clique of size~2}.
The general case is if there are \mbox{$d > 1$} pairwise disjoint unavoidable sets (a \emph{clique of size $d$}); at least $d$ clues are then necessary to hit them all, and thus the union of these $d$ sets is an unavoidable set of degree~$d$.
Hence, after \emph{checker} has found some degree~1 unavoidable sets, all we need to do to obtain a sizeable collection of higher-degree unavoidable sets is search for cliques.
Say $\ell$ unavoidable sets of degree~1 have been found.
Then for each \mbox{$d = 2, \ldots, k$}, where $k$ is again the size of the desired hitting set, we simply consider all \mbox{$\binom{\ell}{d}$} combinations of $d$ unavoidable sets and check for each if the $d$ unavoidable sets in question are pairwise disjoint.
If so, then we have found a clique of size~$d$, and so we just take the union of these $d$ unavoidable sets to produce a degree~$d$ unavoidable set.
In practice, it turns out to be optimum to use cliques of size~2, 3, 4 and 5 only, so that detection of failure does not happen until at least twelve clues have been picked.

It may be surprising that such a simple observation can actually make an important difference, but with \emph{checker}, on the algorithmic side (as opposed to the implementation side), this was really the key ingredient in making the exhaustive search for a 16-clue sudoku puzzle feasible.\footnote{``In a [sudoku] puzzle, what you haven't looked at yet is probably where you're meant to make progress. And in science it's what you haven't looked at yet often times.'' ---Thomas Snyder (a.k.a.\ \emph{Dr.~Sudoku}), three-times sudoku world champion, explaining similarities between solving a sudoku and carrying out scientific research (Source: http://www.youtube.com/watch?v=WMNal53nBtE from 3:25 to 4:18)}
Finally, we refer to these cliques as \emph{trivial} higher-degree unavoidable sets.
Actually, there are also examples of \emph{nontrivial} higher-degree unavoidable sets, e.g., unavoidable sets of degree~2 that are not merely the union of two disjoint unavoidable sets of degree~1, see Section~\ref{subsec:higherdeg}.
Owing to time constraints, however, we did not fully investigate the usefulness of such nontrivial higher-degree unavoidable sets, and in the end these did not play a very important role in the computation.

\subsection{Running all 16-clue candidate puzzles through a sudoku solver}\label{subsec:solver}

In this subsection we explain Step~C3.
We required a sudoku solver module (procedure) for \emph{checker}, because we had to test all hitting sets found for a unique completion, the reason being that we worked with a subcollection of the collection of all unavoidable sets.
To avoid doing something in a mediocre way that others had already excelled at, we did not actually develop our own sudoku solver.
Instead, in the original \emph{checker}, we used the public domain solver \emph{suexk} written by Guenter Stertenbrink \cite{stertenbrink}.
However, for the new version of \emph{checker}, we switched to the open-source solver \emph{BBSudoku}, by Brian Turner \cite{briturner}.
At the time we had to decide which solver to use, Turner's was the fastest one available, to our knowledge.
Note that Stertenbrink's algorithm is based on an exact cover problem solver, while Turner's approach is to solve a sudoku puzzle much like a human would: at every step, the idea is to first look for naked singles, then hidden singles, then locked candidates, then to check for an X-Wing, Swordfish, etc., and finally to make a guess, i.e., perform trial and error.
For the actual computations, we slightly modified Turner's solver.
We removed all of the advanced solving methods in order that the logic would be as simple as possible.
The modified solver consisted of only about 150~statements, yet was capable of checking over 50,000 16-clue trial puzzles for a unique completion per second, a testament to Turner's design.
In fact, not making use of the advanced solving techniques produces the fastest possible version of this solver, as the author also remarks in a comment in the source code.

\section{Some theory of unavoidable sets}\label{sec:unavsets}

In Section~\ref{subsec:unavsets} we gave a brief introduction to unavoidable sets, and outlined how we go about finding them in a given completed sudoku grid.
We said that the algorithm in the original version of \emph{checker} did a lot of redundant work compared with the new \emph{checker}.
We will now make all this precise, and in particular we explain how investigating the theory of unavoidable sets makes it possible to remove the aforementioned redundancy relatively easily.

\subsection{Definition and elementary properties}

Throughout this section, a sudoku solution grid will formally be a function
	$$\{0, \ldots, 80\} \to \{1,\ldots,9\},$$
and when we say ``let $X$ be a subset of a sudoku solution grid $G$'' we identify $G$ with the corresponding subset of the cartesian product \mbox{$\{0,\ldots,80\} \times \{1,\ldots,9\}$}.

\begin{defn}\label{def:UA}
	Let $G$ be a sudoku solution grid.
	A subset \mbox{$X \subset G$} is called an \emph{unavoidable set} if \mbox{$G \backslash X$}, the complement of $X$, has multiple completions.
	An unavoidable set is said to be \emph{minimal} if no proper subset of it is itself unavoidable.
\end{defn}

So if a subset of a sudoku solution grid does not intersect every unavoidable set contained in that grid, then it cannot be used as a set of clues for a sudoku puzzle because there will be more than one solution.
In other words, a set of clues for a valid puzzle \emph{must} contain at least one clue from every unavoidable set.
In fact, the converse is true as well:

\begin{lemma}
Suppose that \mbox{$X \subset G$} is a set of clues of a sudoku solution grid $G$ such that $X$ hits (intersects) every unavoidable set of $G$.
Then $G$ is the only completion of $X$.
\end{lemma}

\begin{proof}
If $X$ had multiple completions, then \mbox{$G \backslash X$} would be an unavoidable set not hit by $X$, which is a contradiction.
\end{proof}

In the original version of \emph{checker}, unavoidable sets in a given sudoku solution grid were found using a straightforward pattern-matching algorithm.
More specifically, \emph{checker} contained several hundred different \emph{blueprints}, where a blueprint is just a representative of an equivalence class of (minimal) unavoidable sets, the equivalence relation again being the one from Section~\ref{sec:EqRel}.
On the online forums, Ed~Russell had investigated unavoidable sets and compiled a list of blueprints, which he sent us.
We added all blueprints of size 12 or less from Russell's list to \emph{checker} (525 blueprints in total).
When actually finding unavoidable sets in a completed sudoku grid, \emph{checker} would compare each blueprint against all grids in the same equivalence class as the given grid, modulo the digit permutations.
That is, \emph{checker} would generate 3,359,232~grids equivalent to the original grid, as explained in Section~\ref{sec:catalog}, and for every grid generated \emph{checker} would try each blueprint for a match.
In total, this yields around 360~unavoidable sets with a typical grid.

Using the algorithm just described, finding unavoidable sets takes approximately 30~seconds per grid.
In 2006 we did not consider this a problem because back then \emph{checker} took over an hour on average to scan a grid for all 16-clue puzzles, so that searching the entire sudoku catalog was completely out of question anyway.
However, once we succeeded in efficiently enumerating the 16-clue hitting sets, in order to make this project feasible we also had to come up with a faster algorithm for finding the unavoidable sets.
While we kept our original pattern-matching strategy, a bit of theory helped us to significantly reduce the number of possibilities to check.
Before we can start developing the theory of unavoidable sets, we need one more definition.

\begin{defn}
	Let $G$ be a sudoku solution grid, and let \mbox{$X \subset G$}.
	We say that a digit $d$, \mbox{$1 \leq d \leq 9$}, appears in $X$ if there exists \mbox{$c \in \{0, \ldots, 80\}$} such that \mbox{$(c,d) \in X$}.
	Similarly we say that a cell $c$ is contained in $X$ if \mbox{$(c,d) \in X$} for some $d$.
\end{defn}

\begin{lemma}\label{lemma:unav}
	Let $G$ be a sudoku solution grid and suppose that \mbox{$U \subseteq G$} is a minimal unavoidable set.
	If $H$ is any other completion of \mbox{$G \backslash U$}, \mbox{$H \neq G$}, then $G$ and $H$ differ exactly in the cells contained in $U$.
	In particular, every digit appearing in $U$ occurs at least twice.
\end{lemma}

\begin{proof}
	If $G$ and $H$ agreed in more cells than those contained in \mbox{$G \backslash U$}, then $U$ would not be minimal as it would properly contain the unavoidable set \mbox{$G \backslash (G \cap H)$}.	
	So when moving between $G$ and $H$, the contents of the cells in $U$ are permuted such that the digits in all cells of $U$ change.
	If there was a digit $d$ contained in only one cell of $U$, that digit could neither move to a different row nor to a different column, since otherwise the row respectively column in question would not contain the digit $d$ anymore at all.
	That is, the digit $d$ stays fixed, in contradiction to what we just noted.
\end{proof}

For our first corollary, recall that a \emph{derangement} is a permutation with no fixed points, i.e., an element \mbox{$\sigma \in S_n$}, \mbox{$n \in \mathbb{N}$}, such that \mbox{$\sigma(k) \neq k$}, for all \mbox{$k = 1, \ldots, n$}.

\begin{cor}\label{cor:derangement}
	Let $G$ be a sudoku solution grid and suppose that \mbox{$U \subseteq G$} is a minimal unavoidable set.
	If $H$ is any other completion of \mbox{$G \backslash U$}, \mbox{$H \neq G$}, then $H$ may be obtained from $G$ by a derangement of the cells in each row (column, box) of $U$.
	Hence the intersection of $U$ with any row (column, box) is either empty or contains at least two elements.
\end{cor}

\begin{proof}
	Follows directly from the last lemma and since the rules of sudoku would be violated otherwise.
\end{proof}

\subsection{Finding minimal unavoidable sets efficiently using blueprints}\label{subsec:blueprints}

The basic idea how to make the actual search for minimal unavoidable sets in a given sudoku solution grid faster is to replace the blueprints from Russell's list by appropriate members of the same equivalence class.
These members are chosen such that only a fraction of grids equivalent to the given one need to be checked for a match.
We will call a blueprint an \mbox{$m \times n$}~blueprint if it hits $m$ bands and $n$ stacks of the \mbox{$9 \times 9$} matrix, and we treat the blueprints according to the number of bands and stacks they hit.
By taking the transpose if necessary, it is no loss of generality to assume that \mbox{$m \leq n$}.

Suppose that \mbox{$m = 1$}, i.e., suppose that a blueprint hits only a single band.
First observe that by Lemma~\ref{lemma:unav}, a \mbox{$1 \times 1$}~blueprint cannot exist, as any digit in a minimal unavoidable set appears at least twice.
So \mbox{$n \geq 2$}.
By swapping bands if necessary, we may then assume that (up to equivalence) only the top band is hit.
Moreover, after possibly permuting some of the rows and/or columns, we may in fact assume that two of the cells of the blueprint are as follows, again because any digit appearing in a minimal unavoidable set occurs at least twice:
\begin{figure}[h]
	\begin{center}
	\includegraphics[scale=0.85]{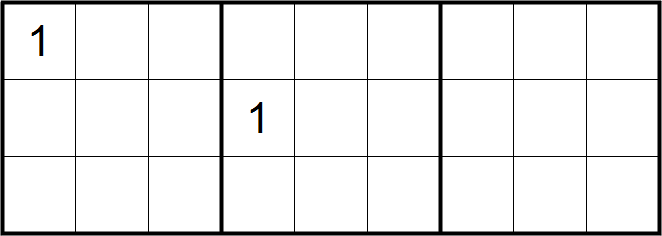}
	\caption{Two digits in any $1 \times n$ blueprint}\label{fig:unav}
	\vspace{-15pt}
	\end{center}
\end{figure}

\noindent For a \mbox{$1 \times 3$} blueprint we further choose, if possible, a representative that has no cells in either the middle or the right column of the right stack.

When actually searching for instances of \mbox{$1 \times 2$}~blueprints in a given grid, i.e., when generating those representatives of the given grid that need to be considered in order to find all occurrences of a \mbox{$1 \times 2$}~blueprint, there will be three possibilities for the choice of top band as well as six permutations of the three rows within the top band (once a band has been chosen as the top band).
So in total there are 18~different arrangements of the rows. 
Compare this to 1,296~row permutations with the original algorithm.
For the columns, a priori, we have to consider all six possible arrangements of the three stacks, as well as all six permutations of the columns in the left stack.
However, the majority of \emph{$1 \times 2$}~blueprints have a stack containing cells in only one column of that stack.
Therefore, if we choose such a stack as the middle stack, then only the left column in the middle stack will contain cells of the blueprint.
So we do not in fact consider all six permutations of the three columns in the middle stack; rather, we try each of the three columns as the left column once only, which means that we use only one (random) arrangement of the two remaining columns in the middle stack.

Of course this will, with probability~\mbox{$0.5$}, miss instances of those blueprints that have digits in the other two columns of the middle stack, which is why all such blueprints are actually contained twice in \emph{checker}, with the respective columns swapped.
Perhaps this is best explained by the following example of a \mbox{$1 \times 2$}~blueprint of size 10, which is saved twice in \emph{checker}'s table of blueprints:

\begin{figure}[h]
	\begin{center}
	\begin{tabular}{lcr}
	\begin{minipage}{6.5cm}
	\begin{center}
	\includegraphics[scale=0.85]{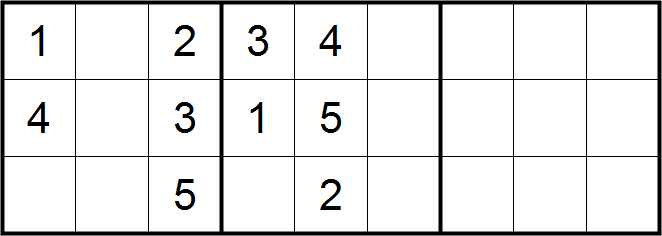}
	\end{center}
	\end{minipage}
	
	&
	\hspace{0.2cm}
	&
	
	\begin{minipage}{6.5cm}
	\begin{center}
	\includegraphics[scale=0.85]{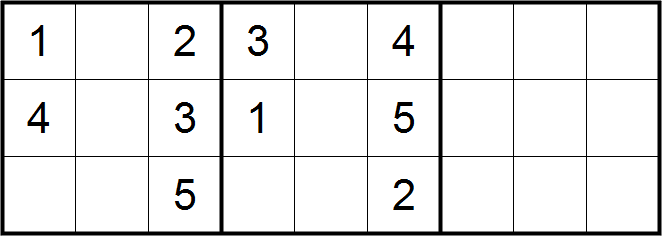}
	\end{center}
	\end{minipage}
	\end{tabular}
	\caption{A $1 \times 3$ blueprint of size~10 that occurs twice in \emph{checker}'s list of blueprints}
	\vspace{-15pt}
	\end{center}
\end{figure}

\noindent
In total, 108~different permutations of the columns will be considered.
However, only in one out of six cases do we actually need to match all the \mbox{$1 \times 2$} blueprints against the corresponding grid, since all our blueprints have the same digit in the top-left cell and in the fourth cell of the second row as explained at the beginning of this subsection (see Figure~\ref{fig:unav}).
Summing up the above discussion, in order to find all \mbox{$1 \times 2$} unavoidable sets in a given sudoku solution grid, instead of having to generate 3,359,232~grids, in actual fact we only need to generate 648~grids.

We search for \mbox{$1 \times 3$} unavoidable sets in a very similar manner; the only additional effort required is that we also need to permute the columns in the right stack of the grid.
Like with the middle stack, we do not actually try all six permutations, but only three permutations --- each of the three columns in the right stack is selected as the left column exactly once.
Since this would again miss half of those unavoidable sets having clues in multiple columns of the right stack, the respective blueprints also appear twice in \emph{checker}.
In particular, \mbox{$1 \times 3$} blueprints having clues in multiple columns of \emph{both} the middle \emph{and} the right stack actually appear four times in \emph{checker}'s table of blueprints, e.g., this one here of size 12:
\begin{figure}[h]
	\begin{center}
	\begin{tabular}{lcr}	
		\begin{minipage}{6.5cm}
			\begin{center}
				\includegraphics[scale=0.85]{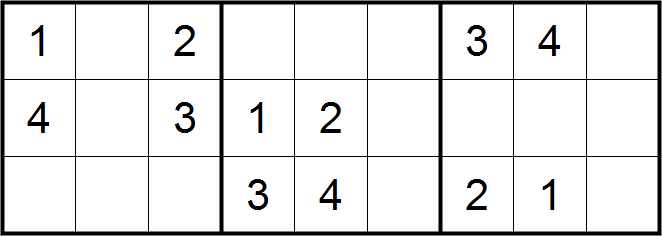}
			\end{center}
		\end{minipage}
	&
		\hspace{0.2cm}
	&
		\begin{minipage}{6.5cm}
			\begin{center}
				\includegraphics[scale=0.85]{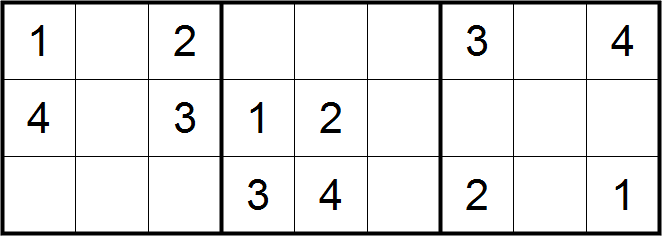}
			\end{center}
		\end{minipage}
	\\
		\vspace{0.2cm}
	\\
		\begin{minipage}{6.5cm}
			\begin{center}
				\includegraphics[scale=0.85]{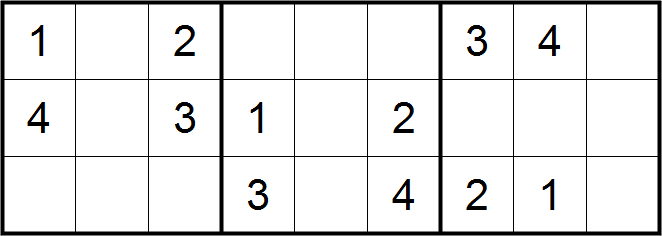}
			\end{center}
		\end{minipage}
	&
		\hspace{0.2cm}
	&
		\begin{minipage}{6.5cm}
			\begin{center}
				\includegraphics[scale=0.85]{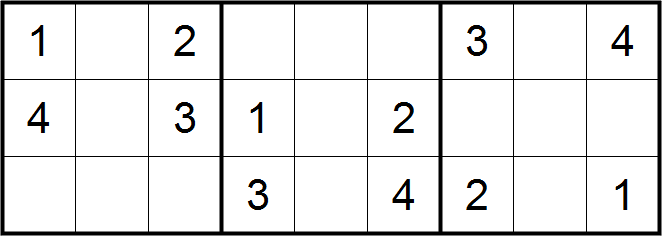}
			\end{center}
		\end{minipage}
	\end{tabular}
	\caption{A $1 \times 3$ blueprint of size~12 that occurs four times in \emph{checker}'s list of blueprints}
	\vspace{-15pt}
	\end{center}
\end{figure}

\noindent
In order to tackle \mbox{$2 \times 2$}, \mbox{$2 \times 3$} and \mbox{$3 \times 3$} blueprints, we need the following result, whose proof is quite long (although not very difficult).

\begin{prop}\label{prop:blueprint}
	Every blueprint is equivalent to one containing the same digit twice in the same band.
\end{prop}
\begin{proof}
First note that for \mbox{$1 \times n$} blueprints there is nothing to prove as the claim follows directly from the fact that any digit in a minimal unavoidable set appears at least twice, see Lemma~\ref{lemma:unav}.
So assume that we are given an \mbox{$m \times n$} blueprint $B$ such that \mbox{$n \geq m \geq 2$}.
We begin by showing that there is either a band or a stack containing the same digit twice.
Let $d$ be any digit occurring in $B$.
Suppose for the sake of contradiction that $d$ does not appear twice in either the same band or the same stack.
Without loss of generality, the only two possibilities for \emph{all} occurrences of $d$ are therefore as follows, recalling that $d$ appears at least twice:

\smallskip
\begin{center}
	\begin{tabular}{lcr}
		\begin{minipage}{5cm}
			\begin{center}
				\includegraphics[scale=0.6]{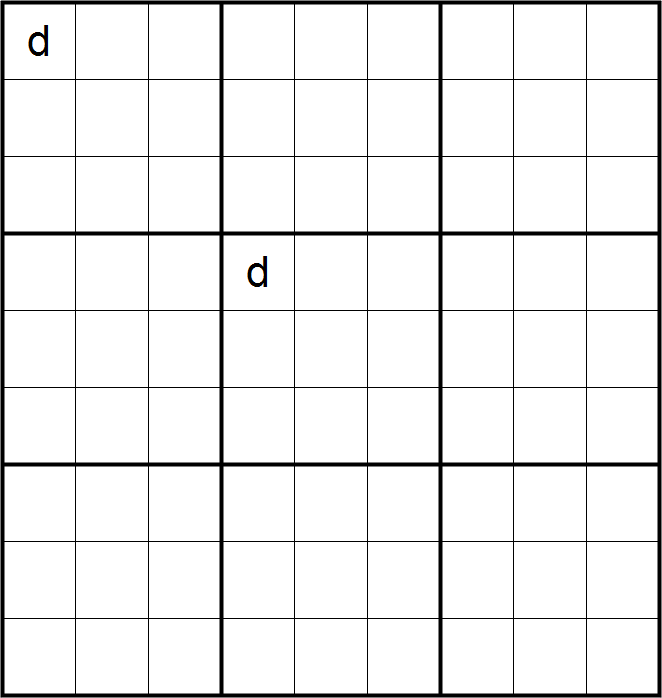}
			\end{center}
		\end{minipage}
	&
		\hspace{0.2cm}
	&
		\begin{minipage}{5cm}
			\begin{center}
				\includegraphics[scale=0.6]{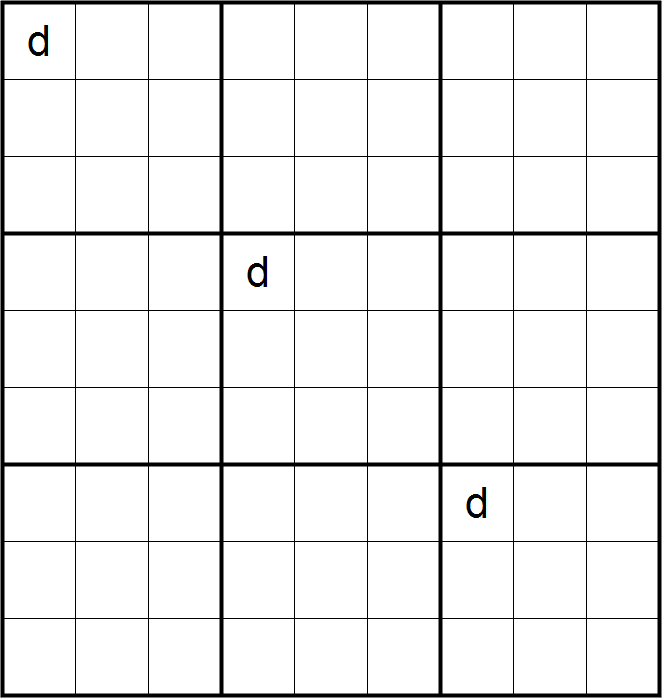}
			\end{center}
		\end{minipage}
	\end{tabular}
\end{center}
By Corollary~\ref{cor:derangement}, if we number the cells in a grid from 0 to 80 inclusive, left-to-right and top-to-bottom, then there exists a permutation \mbox{$\sigma\colon \{0,\ldots,80\} \to \{0,\ldots,80\}$} such that
$$
	\left\{\begin{array}{ll}
		\sigma(c) \neq c~\textrm{ but }~\mathrm{box}(\sigma(c)) = \mathrm{box}(c), & c \in B, \\
		\sigma(c) = c, & c \notin B,
	\end{array}\right.
$$
and whose application to the blueprint $B$ respects the rules of sudoku.
Here, the function $\mathrm{box}$ evaluates to the index of the \mbox{$3 \times 3$} box the given cell belongs to.
(Formally, \mbox{$\mathrm{box}\colon \{0,\ldots,80\} \to \{0,\ldots,8\}$} is defined by
\mbox{$c \mapsto \lfloor c / 27 \rfloor \times 3 + \lfloor c~\!\!\!\pmod{9}/ 3\rfloor$}.)
But then in both of the above cases, by the properties of $\sigma$ just stated, any instance of the digit $d$ would have to stay fixed as the cell containing it is only allowed to move inside its box, but if that cell actually left either its original row or column, the rules of sudoku would be violated, which is the desired contradiction.
In detail, if a cell of $B$ containing the digit $d$ left, e.g., its row under $\sigma$, then the row in question would lack $d$, as any other appearance of $d$ that might compensate for the ``loss'' of the original instance of $d$ in the affected row would have to be contained in the same band, because, after applying $\sigma$, all cells are still in the same box as they originally were.
Similarly if a cell of $B$ containing $d$ left its original column under $\sigma$.

This proves our first claim, that $d$ necessarily appears twice in either the same band or the same stack.
If there is a band containing the digit $d$ twice, we are done.
Otherwise, for a \mbox{$2 \times 2$} or a \mbox{$3 \times 3$} blueprint we simply take the transpose, and again the claim is proven.
We are left with the case that $B$ is a \mbox{$2 \times 3$} blueprint such that every digit appears exactly twice in the same stack --- by the pigeon hole principle, no digit of $B$ can appear three or more times, because then at least one of the bands that $B$ contains cells of would have that digit multiple times.

Without loss of generality, two digits of $B$ are as shown on the left; the respective digits of \mbox{$\sigma(B)$} are then necessarily as shown on the right:
\smallskip
\begin{center}
	\begin{tabular}{lcr}
		\begin{minipage}{5cm}
			\begin{center}
				\includegraphics[scale=0.85]{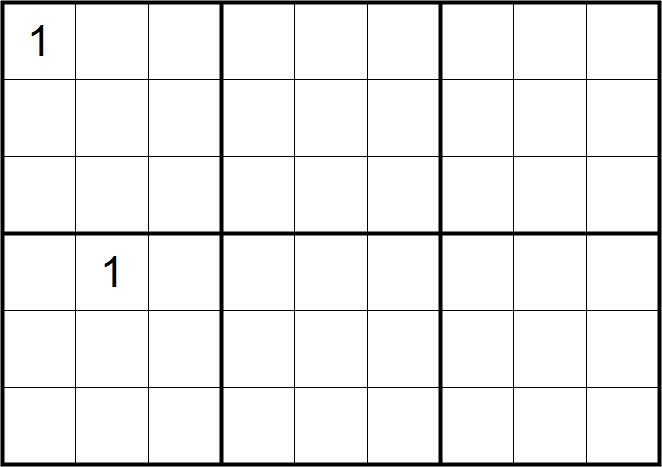}
			\end{center}
		\end{minipage}
	&
		\begin{tabular}{c}
			\small{$\sigma$} \\ $\mapsto$
		\end{tabular}
	&
		\begin{minipage}{5cm}
			\begin{center}
				\includegraphics[scale=0.85]{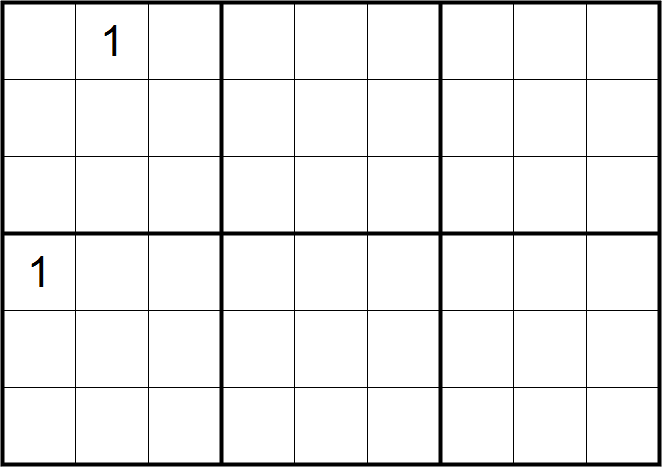}
			\end{center}
		\end{minipage}
	\end{tabular}
\end{center}

\smallskip\noindent
This illustrates a general principle --- when applying $\sigma$ to $B$, \emph{every} cell of $B$ stays in its original row, once again because no digit appears twice in the same band of $B$ and as all cells of $B$ stay in their original boxes under $\sigma$.
Hence if $d$ is any digit appearing in $B$ and if $c$ and $c'$ are the two cells of $B$ containing $d$, then
$$
	\begin{array}{rclrclrcl}
		\mathrm{box}(\sigma(c)) & = & \mathrm{box}(c),&
		\mathrm{row}(\sigma(c)) & = & \mathrm{row}(c),&
		\sigma(c) & \neq & c,\\
		\qquad\mathrm{box}(\sigma(c')) & = & \mathrm{box}(c'),\qquad&
		\mathrm{row}(\sigma(c')) & = & \mathrm{row}(c'),\qquad&
		\sigma(c') & \neq & c'.
	\end{array}
$$
It now follows immediately that
\begin{equation}\label{eq:col}
	\mathrm{col}(\sigma(c)) \;=\; \mathrm{col}(c') \qquad\textrm{and}\qquad \mathrm{col}(\sigma(c')) \;=\; \mathrm{col}(c),
\end{equation}
since the rules of sudoku would be violated otherwise.
(The functions $\mathrm{row}$ and $\mathrm{col}$ are defined in the obvious way.)
We summarize the above observations as follows.

\emph{Every digit appearing in $B$ occurs exactly twice, and both instances are in the same stack; moreover, every cell of $B$ stays in its original row and box under $\sigma$ but moves to the original column of the other cell of $B$ containing the same digit.}

We will complete the proof by showing that this contradicts minimality of $B$ --- recall that any blueprint by definition is an instance of a \emph{minimal} unavoidable set.
In brief, it suffices to notice that the subset of cells of $B$ contained in the left stack, which is a proper subset of $B$ because $B$ also has cells in the other two stacks of the grid, already is itself unavoidable.
We shall now make this precise.
To this end, let \mbox{$G$} be any completion of $B$, and let \mbox{$G^{(0)}$}, \mbox{$G^{(1)}$} and  \mbox{$G^{(2)}$} denote the three stacks of $G$.
That is to say, for each \mbox{$s=0,1,2$}, we think of \mbox{$G^{(s)}$} as a function
$$\mbox{$\{9i + j + 3s \mid i = 0, \ldots, 8, j = 0, 1, 2\}$} \;\to\; \{1,\ldots,9\},$$
and we again identify \mbox{$G^{(s)}$} with the corresponding subset of \mbox{$\{0,\ldots,80\} \times \{1,\ldots,9\}$}.
Recall that in the beginning of the proof we chose $\sigma$ such that \mbox{$\mathrm{box}(\sigma(c)) = \mathrm{box}(c)$}, for all \mbox{$c = 0, \ldots, 80$}.
So when $\sigma$ is applied to $G$, any cell of the grid stays in its original box, in particular, any cell stays in  its original stack.
Therefore, the set \mbox{$\sigma(G^{(s)}) = \{(\sigma(c),d) \mid (c,d) \in G^{(s)}\}$} is still the stack of index $s$ (of a different grid), for each \mbox{$s=0,1,2$}, and so there is no ambiguity to set
$$H \;=\; \sigma(G^{(0)}) \cup G^{(1)} \cup G^{(2)}.$$
Note that $H$ is a valid sudoku solution grid:
\begin{itemize}
\item By the above, no cell of $B$ moves to a different row under $\sigma$ (recall also that cells in \mbox{$G \backslash B$} are fixed by $\sigma$ anyway).
Therefore, if \mbox{$d_1,d_2,d_3$} are the three digits contained in any row of the stack \mbox{$G^{(0)}$}, then the corresponding row of the stack \mbox{$\sigma(G^{(0)})$} will also contain the same three digits \mbox{$d_1,d_2,d_3$}, though possibly in a different order.
In any case, each row of $H$ contains every number from 1 to 9 exactly once.
\item Although some cells of \mbox{$G^{(0)}$} do change columns when $\sigma$ is applied, every column of $H$ nonetheless contains only distinct entries --- as if \mbox{$c \in G^{(0)}$} is such that \mbox{$\mathrm{col}(\sigma(c)) \neq \mathrm{col}(c)$}, then \mbox{$c \in B$} and from (\ref{eq:col}) there exists \mbox{$c' \in G^{(0)}, c' \neq c,$} such that the cells $c$ and $c'$ carry the same digit and \mbox{$\mathrm{col}(\sigma(c')) = \mathrm{col}(c)$}, i.e., $c'$ acts as the replacement for $c$ in the original column of $c$ (and conversely).
\item Because $\sigma$ permutes the cells in each box, it is immediately clear that every box of \mbox{$\sigma(G^{(0)})$} still contains nine different numbers.
\end{itemize}
Observe now that the set \mbox{$U := G \backslash (G \cap H)$} is a proper subset of $B$ that is in fact unavoidable.
First, $U$ is a subset of $B$ as, trivially, \mbox{$G \supseteq G \backslash B$}, and also \mbox{$H \supseteq G \backslash B$} since $\sigma$ fixes cells not contained in $B$.
So \mbox{$G \cap H \supseteq G \backslash B$}, and taking complements we therefore get \mbox{$U \subseteq B$}.
Second, $U$ is a \emph{proper} subset of $B$ because \mbox{$G \cap H \supseteq G^{(1)} \cup G^{(2)}$} but \mbox{$B \cap (G^{(1)} \cup G^{(2)}) \neq \emptyset$}, recalling that $B$ contains cells of the middle and the right stack.
Finally, directly from the definition, $U$ is unavoidable as its complement has two different completions, namely $G$ and $H$.
This completes the proof of Proposition~\ref{prop:blueprint}.
\end{proof}

\noindent
So for \emph{any} blueprint it is no loss of generality to assume that two digits in the top band are as shown in Figure~\ref{fig:unav}, not just for \mbox{$1 \times n$} blueprints.
The actual algorithm used when searching for \mbox{$2\times 2$}, \mbox{$2\times 3$}, and \mbox{$3\times 3$} unavoidable sets is very similar to the one for \mbox{$1 \times n$} unavoidable sets --- the only difference is that we further arrange for each blueprint to be of one of the following three types, again to reduce the number of possibilities to check (note that the first and the third type overlap):

\smallskip
\begin{figure}[h!]
	\begin{center}
	\begin{minipage}{12.8cm}
	\begin{tabular}{rcccl}
		\begin{minipage}{3.5cm}
			\begin{center}
				\includegraphics[scale=0.85]{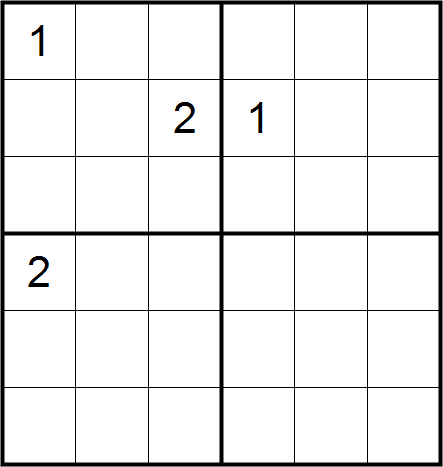}
			\end{center}
		\end{minipage}
	&
		\hspace{0.1cm}
	&
		\begin{minipage}{3.5cm}
			\begin{center}
				\includegraphics[scale=0.85]{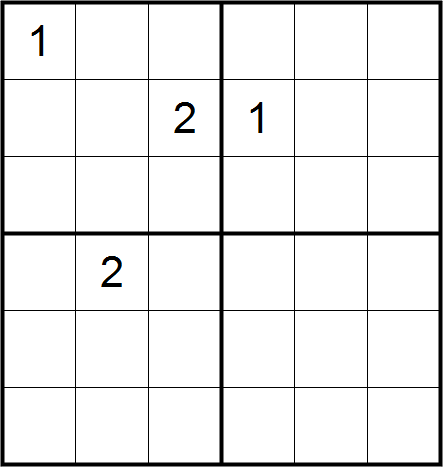}
			\end{center}
		\end{minipage}
	&
		\hspace{0.1cm}
	&
		\begin{minipage}{3.5cm}
			\begin{center}
				\includegraphics[scale=0.85]{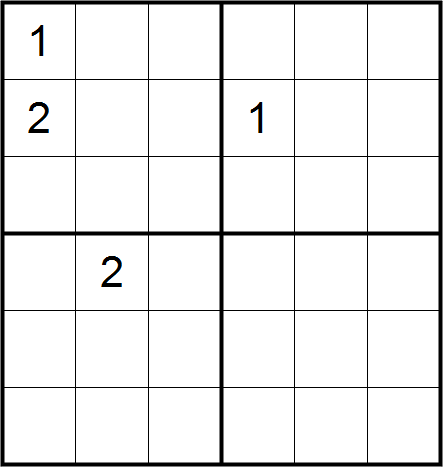}
			\end{center}
		\end{minipage}
	\end{tabular}
	\caption{The three types of \mbox{$m \times n$} blueprint with \mbox{$m \geq 2$} (only the top-left \mbox{$6 \times 6$} subgrid is shown)}
	\vspace{-15pt}
	\end{minipage}
	\end{center}
\end{figure}

\noindent
With the above ideas implemented, finding all unavoidable sets of size up to 12 in a sudoku solution grid takes less than 0.05~seconds on average.
	
\subsection{Higher-degree unavoidable sets}\label{subsec:higherdeg}

This section gives the details of the theory of higher-degree unavoidable sets introduced in Section~\ref{subsec:sketchC2}.
There are unavoidable sets that require more than one clue, which we call \emph{higher-degree} unavoidable sets.
Let us illustrate this with an example.

\begin{figure}[h]
\begin{center}
	\includegraphics[scale=0.85]{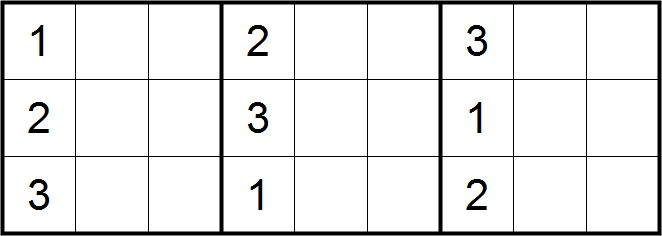}
	\caption{An unavoidable set of degree~2}\label{fig:unav92I}
	\vspace{-15pt}
\end{center}
\end{figure}

\noindent
Note that two clues are needed from the nine digits shown in order to completely determine these nine cells.
Because, if only one is given, the other two digits may be interchanged.
These nine cells form an unavoidable set requiring two clues.
Technically, the above nine clues are really the union of nine minimal unavoidable sets of size six each, and the intersection of these nine minimal unavoidable sets is empty, hence one clue is not enough to hit all of them.
Thus the above is an example of a degree~2 unavoidable set.
If we say that an unavoidable set as defined earlier (see the definition on p.~\pageref{def:UA}) is an unavoidable set of degree 1, then we may recursively define the notion of an unavoidable set of degree \mbox{$k > 1$}.

\begin{defn}
A nonempty subset $U$ of a sudoku solution grid $G$ is said to be an \emph{unavoidable set of degree \mbox{$k > 1$}} if for all \mbox{$c \in U$} the set \mbox{$U \backslash \{c\}$} is an unavoidable set of degree \mbox{$k - 1$}.
\end{defn}

There is the following alternative characterization of higher-degree unavoidable sets.

\begin{lemma}\label{lemma:altchar}
Let $G$ be a sudoku solution grid and let \mbox{$k \in \mathbb{N}$}.
A nonempty subset \mbox{$U \subseteq G$} of size $n$ is unavoidable of degree $k$ if and only if for all $\binom{n}{k-1}$ choices of distinct elements \mbox{$c_1, \ldots, c_{k-1} \in U$}, the set \mbox{$U \backslash \{c_1,\ldots,c_{k-1}\}$} is unavoidable.
\end{lemma}
\begin{proof}
Induction on $k$.
For \mbox{$k = 1$} the result just says that $U$ is unavoidable of degree~1 if and only if $U$ is unavoidable as defined earlier, so there is nothing to prove.
Suppose now that \mbox{$k > 1$} and the result is true for smaller values of $k$.
Assume that $U$ satisfies the hypothesis of the ``if''-direction and let \mbox{$c \in U$} be an arbitrary element; we need to show that \mbox{$U \backslash \{c\}$} is unavoidable of degree \mbox{$k - 1$}.
Using induction it suffices if \mbox{$\mbox{$(U \backslash \{c\}) \backslash \{c_1, \ldots, c_{k-2}\} = U \backslash \{c,c_1, \ldots, c_{k-2}\}$}$} is unavoidable for all possible combinations of distinct \mbox{$c_1, \ldots, c_{k-2} \in U \backslash \{c\}$}, which however is true by assumption.

For the ``only~if''-direction suppose that $U$ is unavoidable of degree $k$ and let \mbox{$c_1, \ldots, c_{k-1} \in U$} be distinct.
By definition, \mbox{$U \backslash \{c_1\}$} is then unavoidable of degree \mbox{$k-1$}, and from induction it thus follows immediately that \mbox{$U \backslash \{c_1, \ldots, c_{k-1}\} \;=\; (U \backslash \{c_1\}) \backslash \{c_2, \ldots, c_{k-1}\}$} is unavoidable, as required.
\end{proof}

As before, we say that an unavoidable set of degree greater than 1 is \emph{minimal} if no proper subset is unavoidable of the same degree.
Furthermore, to ease notation, we will say that $U$ is an \mbox{$(m,k)$} unavoidable set if $U$ is an unavoidable set of degree $k$ having $m$ elements.
So the example in Figure~\ref{fig:unav92I} is a \mbox{$(9,2)$} unavoidable set that is the union of nine \mbox{$(6,1)$} unavoidable sets.
One can easily construct higher-degree unavoidable sets, e.g., the union of any two disjoint degree~1 unavoidable sets is trivially an unavoidable set of degree 2.
More generally, we have the following result.
\begin{prop}
Let \mbox{$U \subset G$} be an \mbox{$(m,k)$} unavoidable set.
If \mbox{$V \subset G$} is an \mbox{$(n,\ell)$} unavoidable set such that \mbox{$U \cap V = \emptyset$}, then \mbox{$U \cup V$} is an \mbox{$(m+n,k+\ell)$} unavoidable set.
\end{prop}
\begin{proof}
The claim follows directly from the last lemma.
\end{proof}
Repeated application of this proposition gives the following useful fact.
\begin{cor}\label{cor:higherunav}
Suppose that \mbox{$U_1, \ldots, U_t$} are unavoidable sets of a sudoku solution grid $G$ of degree \mbox{$k_1, \ldots, k_t$}, respectively.
Assume further that the $U$'s are pairwise disjoint.
Then \mbox{$U_1 \cup \cdots \cup U_t$} is an unavoidable set of degree \mbox{$k_1 + \cdots + k_t$}.
\end{cor}

\begin{defn}
An unavoidable set $U$ of degree \mbox{$k > 1$} is said to be \emph{nontrivial} if there does not exist an unavoidable set $U_1$ of degree $k_1$ and an unavoidable set $U_2$ of degree $k_2$ and disjoint from $U_1$, \mbox{$k_1, k_2 \in \mathbb{N}$}, such that \mbox{$U = U_1 \cup U_2$} and \mbox{$k = k_1 + k_2$}; otherwise, we say that $U$ is \emph{trivial}.
\end{defn}
So the \mbox{$(9,2)$} unavoidable set shown in Figure~\ref{fig:unav92I} is nontrivial.
In actual fact, it is one of only two types of nontrivial \mbox{$(9,2)$} unavoidable sets, the other being this one here:

\begin{figure}[h]
\begin{center}
	\includegraphics[scale=0.85]{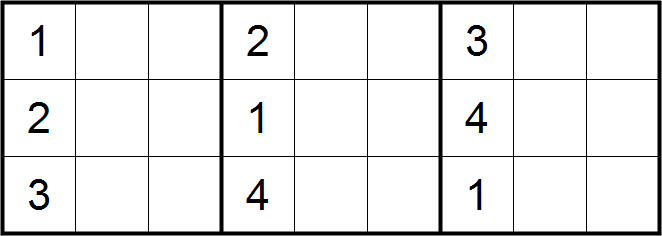}
	\caption{A nontrivial (9,2) unavoidable set}\label{fig:unav92II}
	\vspace{-15pt}
\end{center}
\end{figure}

\noindent
Earlier we said that, so far, nobody has found a purely mathematical proof that 8~clues are not sufficient for a sudoku puzzle to have a unique solution.
However, if one assumes that the sudoku solution grid in question has a nontrivial (9,2) unavoidable set (of either type), then it is actually easy to see why such a grid cannot contain a proper 8-clue puzzle.

We classified all minimal \mbox{$(m,2)$} unavoidable sets for \mbox{$m \leq 11$}.
The result was that no \mbox{$(m,2)$} unavoidable sets exist for \mbox{$m \leq 7$}.
While \mbox{$(8,2)$} unavoidable sets do exist, all these are trivial, i.e., any \mbox{$(8,2)$} unavoidable set is the union of two disjoint \mbox{$(4,1)$} unavoidable sets.
Similarly, minimal \mbox{$(10,2)$} unavoidable sets exist, but again, all these are trivial, i.e., the disjoint union of a \mbox{$(4,1)$} and a \mbox{$(6,1)$} unavoidable set.
There are seven distinct types of nontrivial minimal \mbox{$(11,2)$} unavoidable sets, which however we did not use in this project --- observe that any minimal \mbox{$(11,2)$} unavoidable set is necessarily nontrivial as there are no minimal \mbox{$(m,1)$} unavoidable sets for \mbox{$m=1,2,3,5,7$}.
Naturally we have the following result.

\begin{prop}\label{prop:higherdegunav}
	Let \mbox{$U \subset G$} be an \mbox{$(m,k)$} unavoidable set.
	Then we need to add at least $k$ elements from $U$ to \mbox{$G \backslash U$} in order to obtain a sudoku puzzle with a unique completion.
\end{prop}

\begin{proof}
	If we add only up to \mbox{$k-1$} clues from $U$ to the puzzle \mbox{$G\backslash U$}, then there will still be multiple completions, since by Lemma~\ref{lemma:altchar} and the assumption that $U$ is an unavoidable set of degree $k$, the set \mbox{$U \backslash \{c_1, \ldots, c_{k-1}\}$} is unavoidable for any choice of \mbox{$c_1, \ldots, c_{k-1} \in U$}.
\end{proof}

This last proposition is really the key result on the theory of unavoidable sets as far as the problem of making the enumeration of hitting sets more efficient is concerned, see Section~\ref{sec:newalgo}.

\subsubsection*{Example}

We close this section with the following example of a sudoku grid that requires 18 clues, at least.
We can prove this fact purely mathematically using unavoidable sets of degree~2.

\begin{figure}[h]
\begin{center}
	\includegraphics[scale=0.85]{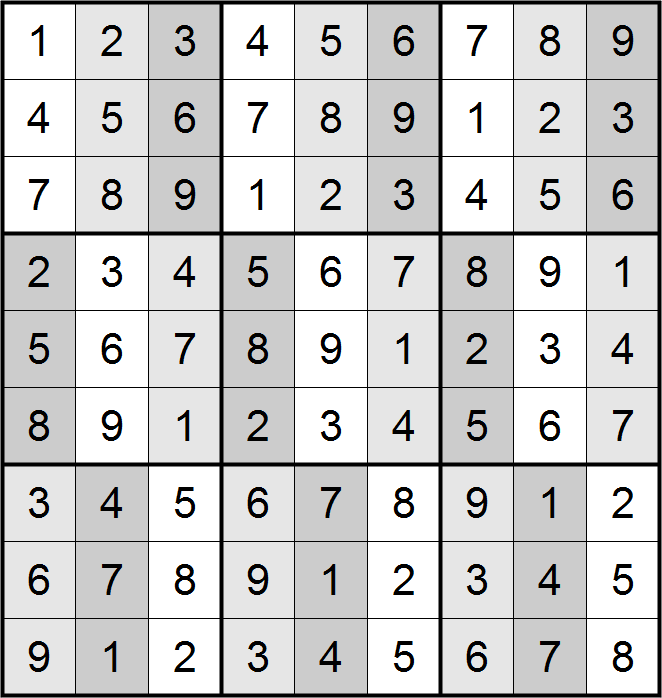}
	\caption{A grid requiring at least 18 clues}\label{fig:canonical}
	\vspace{-15pt}
\end{center}
\end{figure}

\begin{lemma}
	Any sudoku puzzle whose solution is this grid has at least $18$ clues.
\end{lemma}

\begin{proof}
	Upon inspection, the grid is the union of nine pairwise disjoint \mbox{$(9,2)$} unavoidable sets, as indicated by the different shadings.
	From Corollary~\ref{cor:higherunav}, this grid is therefore an \mbox{$(81,18)$} unavoidable set and thus requires at least 18 clues for a unique solution, by the last proposition.
\end{proof}

We note that some choices of 18 clues lead to a unique solution, but some do not.
We also remark that results like this lemma are impossible to prove by hand for a general grid.
It is possible to prove the lemma by hand only because this grid is highly structured.

\section{Hitting set algorithm from the original checker}\label{sec:orgalgo}

In this section we will explain how the original (2006) version of \emph{checker} enumerated the hitting sets for a given family of unavoidable sets in a sudoku solution grid.
Section~\ref{very_basic} describes how we imple\-mented the obvious hitting set algorithm, and Section \ref{subsec:deadcluevec} explains the one improvement we made.
Formally, the hitting set problem is as follows: given a triple \mbox{$(U,\mathcal{F},k)$}, where $U$ is a finite set (the \emph{universe}), \mbox{$\mathcal{F} \subset 2^U$} is a family of subsets of $U$, and $k$ is a positive integer, the task is to find all subsets \mbox{$H \subset U$}, \mbox{$\#H = k$}, such that \mbox{$H \cap S \neq \emptyset$}, for all \mbox{$S \in \mathcal{F}$}.
Each set $H$ is called a \emph{hitting set} for $\mathcal{F}$.

The hitting set problem is NP-complete \cite{karp}.
And indeed, the algorithm presented in this section, despite already having the advantage over a standard backtracking algorithm that each hitting set is enumerated once only, would not have been fast enough to solve the sudoku minimum number of clues problem using a reasonable amount of compute time until about the year 2020.
Only with the improvements described in Section~\ref{sec:newalgo} was this project feasible in 2011.
We should also mention that a common way to approach the hitting set problem is actually a greedy algorithm.
However, as the result a greedy algorithm produces is only approximate there would have been no sense applying such an algorithm to the sudoku minimum number of clues problem.

\subsection{Description of the basic algorithm}\label{very_basic}

For our original version of \emph{checker} from 2006 we essentially used the obvious algorithm for finding hitting sets, with one small, but powerful, improvement.
The strategy of the basic algorithm may be described in one paragraph.
Given a completed sudoku grid as well as a collection of unavoidable sets for that grid, candidate 16-clue puzzles are constructed by recursively adding clues from unavoidable sets.
Whenever the next clue is added to the puzzle being built up to, one first finds an unavoidable set that does not contain any of the clues picked so far and then branches in all possible ways.
(So if that unavoidable set has $d$ clues, there will be $d$ branches.)
Repeat until 16~clues have been picked.
If all unavoidable sets in the given collection have been hit before the $16^{th}$ clue is reached, the remaining clues needed are added in all possible ways.
It is also understood that at each stage, in order to minimize the number of branches, one always chooses a smallest unavoidable set not yet hit.

The data structure we used to accomplish the above was as follows.
Suppose that there are $m$ members in our family $\mathcal{F}$ of unavoidable sets.
Say \mbox{$\mathcal{F} = \{U_0, \ldots, U_{m-1}\}$} where each \mbox{$U_i \subset \{0, \ldots, 80\}$} and \mbox{$\# U_i \leq \# U_{i+1}$}.
Then for every clue \mbox{$c \in \{0,\ldots,80\}$} there is a binary vector of length $m$, called the \emph{hitting vector for $c$} and which we will denote $\mathrm{hitvec[c]}$, whose $i^{th}$ slot is given by \mbox{$\mathcal{X}_{U_i}(c)$}, where for any subset \mbox{$S \subset \{0,\ldots,80\}$} the function \mbox{$\mathcal{X}_S\colon \{0,\ldots,80\} \to \{0,1\}$} is defined by
$$\mathcal{X}_S(s) = \left\{\begin{array}{cl}
															1, & \mbox{$s \in S,$} \\
                              0, & \mbox{$s \notin S.$}
                            \end{array}\right.$$
So $\mathcal{X}_S$ is just the usual \emph{characteristic function} (or \emph{indicator function}) for $S$ with respect to \mbox{$\{0,\ldots,80\}$}, and the $i^{th}$ slot of \mbox{$\mathrm{hitvec[c]}$} records whether or not the clue $c$ is contained in the unavoidable set \mbox{$U_i$}.
Here is pseudocode for the procedure that sets up the hitting vectors:\label{page:InitHittingVectors}
\begin{verbatim}
 InitHittingVectors(U,m)    // U = array of unavoidable sets, m = number of sets
   create array hitvec[0..80]
   for c from 0 to 80
     do hitvec[c] := (X_{U[0]}(c),...,X_{U[m-1]}(c))
\end{verbatim}
We store the hitting set being constructed in the array $\mathrm{hitset}$.
When enumerating the hitting sets, we need to keep track of which unavoidable sets have been hit already.
Therefore, there is another array of binary vectors of length $m$, \mbox{$\mathrm{statevec[0..16]}$}.
Initially we set \mbox{$\mathrm{statevec[0]} := (0,\ldots,0)$}, i.e., \mbox{$\mathrm{statevec[0]}$} is the zero vector.
If we add the clue $c$ at the $j^{th}$ step, \mbox{$0 \leq j \leq 15$}, then we simply set \label{page:UpdateHittingVectors}
\begin{verbatim}
 hitset[j+1]   := c
 statevec[j+1] := statevec[j] OR hitvec[c]
\end{verbatim}
That is, we first save the clue $c$ to the array $\mathrm{hitset}$ and then perform componentwise (bitwise) boolean $\mathrm{OR}$ on the binary vectors $\mathrm{statevec[j]}$ and $\mathrm{hitvec[c]}$ and store the result in $\mathrm{statevec[j+1]}$.
Therefore $\mathrm{statevec[j+1]}$ contains a 1 in the $i^{th}$ slot if and only if either $\mathrm{statevec[j]}$ or $\mathrm{hitvec[c]}$ has a 1 in the $i^{th}$ slot, which is so if and only if either the set $U_i$ was already hit, or if \mbox{$c \in U_i$}.
We recursively do this until \mbox{$j = 16$} or \mbox{$\mathrm{statevec[j]}=(1,\ldots,1)$}.
In the latter case, i.e., if at some stage all unavoidable sets in our collection have been hit, we add \mbox{$16 - j$} more clues to the hitting set in all possible ways.

\subsection{Using the dead clue vector to prevent multiple enumeration of hitting sets}\label{subsec:deadcluevec}

We describe here the one improvement (over the obvious hitting set algorithm) that we used in the original release of \emph{checker}.
It addresses one clear shortcoming of the above, naive, algorithm, namely the fact that this algorithm will enumerate most candidate 16-clue puzzles multiple times.
For example, suppose that the first three sets in our family $\mathcal{F}$ of unavoidable sets to be hit are
\begin{eqnarray*}
	U_1 & = & \{0,3,9,12\}, \\ 
	U_2 & = & \{0,1,27,28\}, \\
	U_3 & = & \{3,4,66,67\}.
\end{eqnarray*}
When we choose \mbox{$0 \in U_1$} as the first clue of the hitting set under construction, then $U_2$ is also hit automatically, so we will use $U_3$ as the set for drawing the second clue from.
The first element in $U_3$ is 3, so one possibility for the first two clues is \mbox{$\{0,3\}$}.
On the other hand, when the algorithm later chooses \mbox{$3 \in U_1$} as the first clue of the hitting set, then $U_2$ is still unhit, so will be used for drawing the second clue from.
However \mbox{$0 \in U_2$}, so again one possibility for the first two clues is \mbox{$\{0,3\}$}.
The following illustration shows the search tree for the example just given:

\begin{figure}[h]
	\begin{center}
	\includegraphics[scale=0.15]{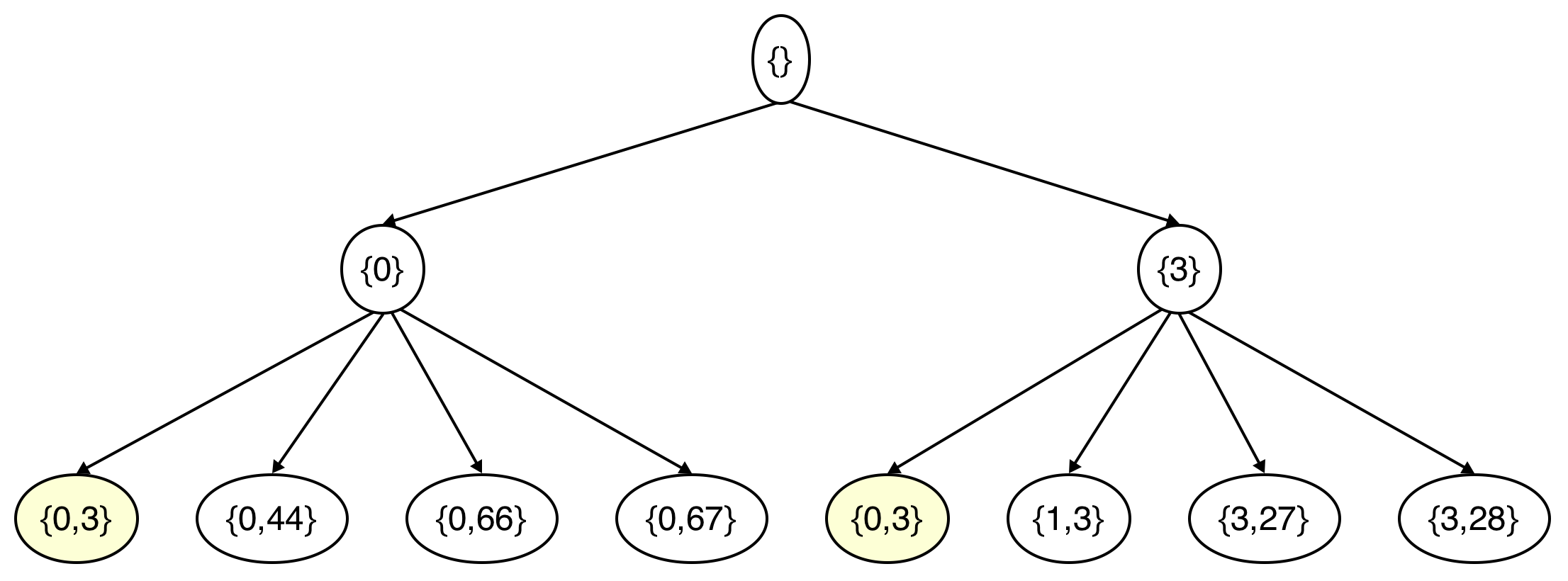}
	\caption{Sample search tree produced by a standard backtracking algorithm}\label{fig:tree}
	\vspace{-15pt}
	\end{center}
\end{figure}

\noindent
To avoid duplicating work as in the example just given, we incorporated the \emph{dead clue vector} into the original \emph{checker} to ensure that every 16-clue puzzle is enumerated once only.
The basic idea is that, whenever we add a clue to the hitting set from an unavoidable set, we consider all smaller clues from that unavoidable set as \emph{dead}, i.e., we exclude these smaller clues from the search (in the respective branch of the search tree only).
We again use a binary vector, of length 81, to keep track of which clues are dead, and whenever we add a clue, we check immediately beforehand whether or not that clue was already excluded earlier.
Going back to the example just given, when we choose \mbox{$3 \in U_1$} as the first element of the hitting set, then we also mark the element \mbox{$0 \in U_1$} as dead.
The next unavoidable set not hit is $U_2$, and normally we would now try 0 as the second clue in the hitting set.
However we excluded 0 as a possible future clue in the previous step, so there will be only three branches in this case.
Only 1, 27, and 28 are actually tried as the second clue of the hitting set being constructed.
The pseudocode for the procedure that initially sets up the required binary vectors is as follows:
\begin{verbatim}
 Initialize(U,m)            // U = array of unavoidable sets, m = number of sets
   for i from 0 to m-1      // first set up vectors needed for the dead clues
     do for each c in U[i]
          do U[i].deadclues[c] := (0,...,0)
             for each d in U[i]
               do if d <= c
                    then U[i].deadclues[c].SetBit(d)
   create arrays deadvec[0..16] and statevec[0..16]
   deadvec[0] := statevec[0] := (0,...,0)
   InitHittingVectors(U,m)    // finally set up hitting vectors as in Section 6.1
\end{verbatim}
The actual procedure that recursively adds clues to a hitting set then looks like this:
\begin{verbatim}
 AddClues(j,U,m,hitset,statevec,deadvec)
   if statevec[j] = (1,...,1)
     then GeneratePuzzles(j,hitset,deadvec[j]) // all unav'ble sets are hit, add
                                              // 16-j clues in all possible ways
   		 return
   else if j=16
      then return  // after drawing 16 clues some unav'ble sets are still not hit
   i := statevec[j].GetIndexOfLowestZeroSlot() // pick first unav'ble set not hit
   for each c in U[i]
     do if deadvec[j].GetBit(c) = 0 // first verify that this clue is still alive
       then hitset[j+1] := c        // add clue to hitting set
            statevec[j+1] := statevec[j] OR hitvec[c]    // update state vector
                                                         // exclude smaller clues
            deadvec[j+1] := deadvec[j] OR U[i].deadclues[c]
            AddClues(j+1,U,m,hitset,statevec,deadvec)    // add more clues
\end{verbatim}
We now convince the reader that the above algorithm does not miss any hitting sets.
Suppose that $G$ is a sudoku solution grid containing a proper 16-clue puzzle $P$.
We need to show that algorithm just presented will find $P$.
Let $U$ be the member of the family $\mathcal{F}$ of unavoidable sets used for drawing the first clue from.
Since $P$ is a proper puzzle, it intersects $U$, so we may set \mbox{$c = \mathrm{min}\;P \cap U$}, i.e., we let $c$ be the smallest clue of $P$ contained in $U$.
When we add $c$ to our candidate hitting set, only clues smaller than $c$ will be excluded from the search, however, as $c$ is the smallest clue of $P$ contained in $U$, no clues of $P$ will actually be excluded.
The exact same is true for the second, third, etc., clue we add --- at each stage, when we add the smallest clue of $P$ also contained in the unavoidable set in question, only clues not appearing in $P$ will be marked `dead'.
After adding the $16^{th}$ clue in this way, our hitting set will equal $P$.
Since $P$ hits all unavoidable sets, $P$ hits all unavoidable sets in the family $\mathcal{F}$ that \emph{checker} uses, and therefore no further unavoidable sets are available, so that \emph{checker} will test the set $P$ for a unique completion (see Section~\ref{sec:checker}), and thus find the 16-clue puzzle $P$.

As a closing remark, we originally added the dead clue vector to \emph{checker} because we wanted to search this special sudoku grid
\begin{figure}[h]
	\begin{center}
	\includegraphics[scale=0.85]{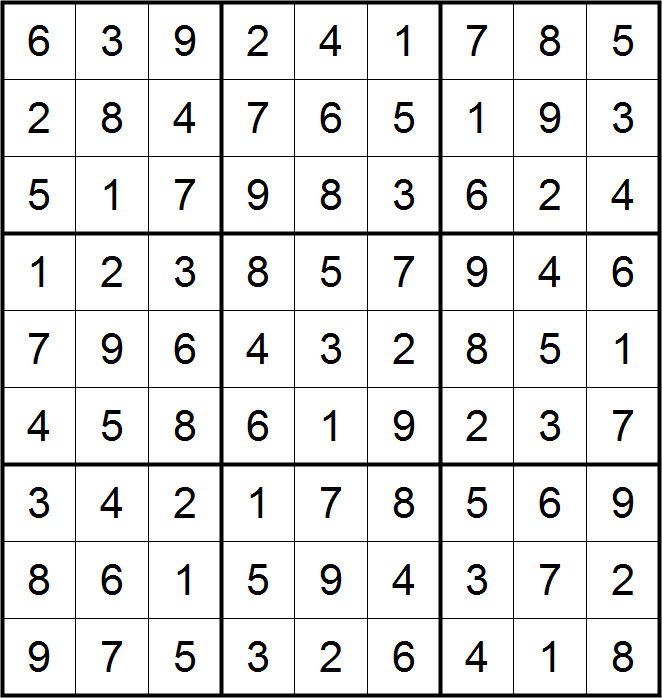}
	\caption{A grid containing 29 proper 17-clue puzzles}\label{img:sfgrid}
	\vspace{-15pt}
	\end{center}
\end{figure}

\noindent
for all 17-clue puzzles.
To this day, this grid holds the record as the grid having the largest known number of 17-clue puzzles (29, all found by Gordon~Royle).
In 2005, this grid was considered a likely candidate to have a 16-clue puzzle, but using \emph{checker} we were the first to prove that there is no 16-clue puzzle in this grid \cite{delahaye}.
Of course, we also wanted to know exactly how many 17-clue puzzles it contained, but the very first \emph{checker} would have taken several months of CPU time to answer that question.
After we had implemented the dead clue vector in 2006, we were able to exhaustively search this grid in less than a week for all 17-clue puzzles.
The result was that Royle had already found all of them, i.e., it was now known that there are \emph{exactly} 29 proper 17-clue puzzles contained in this grid.

\section{Hitting set algorithm of the new checker}\label{sec:newalgo}

This section is a detailed version of the outline provided in Section \ref{subsec:sketchC2}.
We will explain our new algorithm, the algorithm that we used to enumerate the hitting sets of size 16 given a family of unavoidable sets for a particular fixed sudoku solution grid.
As far as this hitting set algorithm is concerned, there are really three improvements over the original \emph{checker} described in the last section.
In this section we describe these improvements.

The first improvement is that we added higher-degree unavoidable sets to \emph{checker} so as to obtain an early ``no'' during the enumeration of hitting sets whenever possible, i.e., in order to abandon the search of a branch as soon as possible.
For instance, if, after drawing 15 clues, there is an unavoidable set of degree~2 that is not yet hit, then we do not have to continue and draw the $16^{th}$ clue as we know that at least two more clues are required for a hitting set.
We outlined this improvement in Section~\ref{subsec:sketchC2}, and we give the details here in Section~\ref{subsec:backtrackbyhigherdeg}.

The second difference is that we discard all those unavoidable sets that have been hit after drawing the first few clues, so that, when adding the remaining clues, we are working with shorter vectors (i.e., a smaller amount of data).
For instance, initially we begin with (up to) 384 minimal unavoidable sets, and after drawing the first seven clues, we check which unavoidable sets have been hit and continue with only the smallest (up to) 128 unavoidable sets.
So when picking the last nine clues, for tracking the minimal unavoidable sets we are using binary vectors of length 128 only, not binary vectors of length 384 as with the first seven clues.
Similarly for the higher-degree unavoidable sets.
We describe this technique in Section~\ref{consolidation}.

The third improvement is that, when choosing which unavoidable set to use for drawing the next clue from, we now invest some effort to make the best, or at least a better, choice.
Recall that with the original \emph{checker}, we selected the unavoidable set to use for drawing the next clue from in a greedy fashion --- we simply used one of smallest size.
However, this is not generally optimum.
A different unavoidable set of the same size, or even a bigger set, may be a better choice since some of its clues may have been excluded from the search already, so that its \emph{effective} (or \emph{real}) size, and hence the number of branches to be taken, may actually be smaller.
Therefore, when choosing unavoidable sets for drawing clues from, the new \emph{checker} also takes the dead clue vector into account.
We describe this in Section~\ref{effective}.

The first of the above changes --- the idea to use higher-degree unavoidable sets --- is certainly the most important one, and without it this computation would not have been feasible for several years.
However, the other two changes, too, saved us a considerable amount of CPU cycles.

\subsection{Improving backtracking using higher-degree unavoidable sets}\label{subsec:backtrackbyhigherdeg}

Here we will explain the most important of the three improvements we made to our hitting set algorithm compared to the version of 2006.
It is about how the use of higher-degree unavoidable sets enables us to considerably prune the search tree.
We begin with the following definition.

\begin{defn}
A collection of pairwise disjoint unavoidable sets in a sudoku solution grid is called a \emph{clique}.
If there is no clique having a greater number of unavoidable sets, then we further say that the clique is \emph{maximum}.\footnote{This terminology comes from graph theory --- in the original \emph{checker}, a maximum clique was found by setting up an undirected graph whose vertices were the minimal unavoidable sets, and where two vertices were adjacent if the corresponding unavoidable sets were disjoint. Hence the term ``max clique number'', or \emph{MCN} for short --- the biggest number of pairwise disjoint unavoidable sets that a grid possesses. In particular, a grid whose MCN is $m$ cannot have a puzzle with fewer than $m$ clues.}
\end{defn}

Recall that, by Proposition~\ref{prop:higherdegunav}, if after drawing $j$ clues, there is an unavoidable set of degree \mbox{$17 - j$} that is not hit, then we do not need to traverse the respective branch of the search tree as we already know that it cannot contain any proper 16-clue puzzles.
On the other hand, from Corollary~\ref{cor:higherunav}, the union of the sets in a clique of size $d$ is an unavoidable set of degree $d$.
Therefore, right before we begin the enumeration of hitting sets, we obtain a (usually quite good) collection of unavoidable sets of degree~2, 3, 4, 5, simply by finding cliques of size~2, 3, 4, 5.

We track these higher-degree unavoidable sets during the enumeration of hitting sets just like the ordinary (degree~1) unavoidable sets, i.e., through the use of state and hitting vectors for each degree.
After twelve clues have been drawn, if there is an unavoidable set of degree~5 in our collection that is not hit, then we may abandon the search and backtrack immediately.
Similarly if there is an unavoidable set of degree~4, 3, or 2 in our collection that is not hit after drawing 13, 14, or 15~clues, then we may abandon the search and backtrack immediately.

This may seem like an obvious way to prune the search tree with the hitting set problem.
However, back in 2008 when we first realized that this idea would allow us to dramatically speed up \emph{checker}, this was not yet described anywhere in the literature.
This is surprising because, like the other improvements we made, it is not at all specific to sudoku but applies in an equal manner to the general hitting set problem.
The first public mention of higher-degree unavoidable sets, to our knowledge, was in a posting of $23^{rd}$ July 2010 to the sudoku programmers' forum by Mladen Dobrichev, who had just released the first version of his open-source tool \emph{GridChecker} \cite{dobrichev}:
\begin{quote}
``UA set is a region of the grid where we know at least one clue must exist. [...]
Additionally there are regions where at least two clues must exist.
A trivial example of such region is the union of 2 mutually disjoint UA sets --- UA sets which have no cell in common.
But, it is not necessary such regions to consist of disjoint UA.
For example 3 UA of size 6 could form region of size 9 requiring at least 2 clues. [...]
Similarly there are regions where at least 3, 4, 5, etc.\ clues must exist.''
\end{quote}

\noindent
It is remarkable how Dobrichev even used the term \emph{trivial} unavoidable set.
However, although \emph{GridChecker} does use higher-degree unavoidable sets, at the time of our work it used a relatively limited collection of such sets, namely those coming from the members of a maximum clique.\footnote{\emph{GridChecker} also uses a maximum clique, however, it does so in a more clever way than our original \emph{checker}.
Moreover, though a powerful collection of unavoidable sets of degree~2 is actually being used, it seems that it is only fully deployed in the method $\mathrm{chunkProcessor \colon\!\!\!\colon\!\! iterateClue}$, for which it was ``rare'', in the words of Dobrichev, that the last clue was being picked there.
(In \emph{GridChecker}, the last clue is usually drawn, if at all necessary, in $\mathrm{chunkProcessor \colon\!\!\!\colon\!\! iterateClueBM}$.)}

There was another public mention of this idea and how it may be used to prune the search tree, in November 2010, when Hung-Hsuan Lin and I-Chen Wu published the paper \cite{linwu} (for an updated version, see \cite{wu}).
In this work, in Lemma~2 it is shown that the search for a $k$-clue puzzle in a sudoku grid may be stopped after selecting $j$ clues if there are at least \mbox{$k - j + 1$} unavoidable sets that are not yet hit (``active'' unavoidable sets, in the language of that paper) and which are pairwise disjoint.
The authors further describe how they used this observation to speed up our original \emph{checker} by a factor of 129 and hence achieved a running time of 13.9~seconds per grid on average.

As far as the problem of finding a clique of a certain size of active unavoidable sets is concerned, they point out that the maximum clique problem is itself NP-complete (like the hitting set problem), and that they therefore use a greedy algorithm for attempting to construct cliques of the desired size.
This is not the most efficient way to construct cliques, however, and it is likely to be the main reason why our own, new \emph{checker} is about twice as fast as the \emph{checker} written by Lin and Wu.
For, constructing cliques over and over again, even just small ones, means duplicating effort.
In contrast, with our new \emph{checker} we compute a large number of cliques of all sizes less than or equal to 5 \emph{exactly once} at the beginning of the search, and keep track of which ones are hit (become inactive) as we add clues.

The consequence is that, e.g., after drawing twelve clues, we merely have to do a boolean $\mathrm{OR}$ of two binary vectors and then check if the resulting vector has a 1 in every slot in order to find out if there is a clique of size~5 not yet hit.
Note that, the moment we find that one slot has a 0, we do not need to compute the remaining slots.
In other words, it is actually sufficient to do the boolean $\mathrm{OR}$ on just part of the vectors involved at first.
In the case of the degree~5 unavoidable sets, we compute the required vector in three steps, checking for a slot containing a zero at the end of each step.
All that can be done efficiently using SIMD programming, whereas constructing a clique of size~5 from scratch is certainly more work and in particular involves more dependencies (where an operation requires the output of the previous one) and is therefore not very suitable for SIMD programming.
Of course, with our method, more work has do be done upfront (while the first eleven clues are picked), but on the other hand, a greedy algorithm will often miss cliques of the required size even though they exist.
On balance, it seems that ours is the more efficient approach.

\subsubsection*{Summary}

Now is a good time to summarize exactly what we do.
We will walk through the case of the cliques of size 4; the other ones are similar.
So suppose that there are $m$ sets \mbox{$\mathrm{U[0]},\ldots,\mathrm{U[m-1]}$} in our initial family of unavoidable sets.
We add the following statements to the procedure $\mathrm{InitHittingVectors}$, see the pseudocode in Section~\ref{very_basic}:
\begin{verbatim}
 var   count := 0;
 const START := 27;
 create array CLQ4[0..32767]
 for i from START to m-1
   do for j from START-1 to i-1
        do if U[i] intersects U[j]
             then continue
           for k from START-2 to j-1
             do if U[k] intersects either U[i] or U[j]
                  then continue
                for l from START-3 to k-1
                  do if U[l] intersects none of U[i],U[j],U[k]
                       then // found a clique of size 4
                         CLQ4[count] := union(U[i],U[j],U[k],U[l])
                         inc(count)
                         if count = 32768
                           then goto SetUpHittingVectors
SetUpHittingVectors:
 create arrays quadhitvec[0..80] and quadstatevec[0..80]
 for c from 0 to 80
   do quadhitvec[c] := (X_{CLQ4[0]}(c),...,X_{CLQ4[count-1]}(c))
 quadstatevec[0] := (0,...,0) // initially, no clique of size 4 is hit
\end{verbatim}
Here $\mathrm{CLQ4}$ is an array of sets, $\mathrm{quadhitvec}$ is an array of binary vectors of length 32,768, and
	$$\mathcal{X}_{\mathrm{CLQ4[i]}}\colon\{0,\ldots,80\} \to \{0,1\}$$
is again the characteristic function for the set \mbox{$\mathrm{CLQ4[i]}$} with respect to \mbox{$\{0,\ldots,80\}$}, see Section~\ref{very_basic}.
The purpose of the constant $\mathrm{START}$ is to ensure that we do not collect cliques that will be hit anyway after drawing 13 clues --- obviously there is no point using cliques involving \mbox{$\mathrm{U}[0]$}, for example, since the first unavoidable set in our family will always be used for drawing the first clue from, so it will certainly be hit by the time we check if there is an unhit clique of size 4.

When adding clues to the hitting set under construction, we need to add one statement (see the pseudocode for the procedure $\mathrm{AddClues}$ in Section~\ref{subsec:deadcluevec}), namely
\begin{verbatim}
 quadstatevec[j+1] := quadstatevec[j] OR quadhitvec[c]
\end{verbatim}
Recall that $j$ is the index of the clue we add to the hitting set and $c$ is the clue itself.
Finally, we need to add logic so as to check if there are any cliques of size 4 not hit after we have drawn the $13^{th}$ clue:
\begin{verbatim}
 if j = 13 and not quadstatevec[13] = (1,...,1)
   then return // a clique of size 4 is unhit after 13 clues have been added
\end{verbatim}
The $\mathrm{return}$-statement means that we backtrack.

\subsection{Consolidating binary vectors for the innermost loops}\label{consolidation}

The second improvement we made to the hitting set algorithm is that, with the unavoidable sets, after picking the first few clues, we discard those sets that have already been hit.
The reason for this is that we then have to carry fewer data, i.e., we can use shorter vectors in the innermost loops of \emph{checker}, where most of the compute time is spent.
For instance, when enumerating hitting sets, initially we use a binary vector of length 32,768 to track the unavoidable sets of degree~4, but after the first five clues have been drawn, we switch to a vector of length 1,536.
Of course we also have to update the respective hitting vectors --- we refer to this process as \emph{consolidating} the hitting vectors.
However, this is not particularly advanced never mind original; what we are doing here is really just a standard \emph{gather} operation.
The key part was to realize that consolidating (gathering) the hitting vector actually helps here.

Continuing with the example of the degree~4 unavoidable sets from Section~\ref{subsec:backtrackbyhigherdeg}, in detail, what we do is as follows.
After five clues have been added to the candidate hitting set, most of the unavoidable sets of degree~4 will usually have been hit, because the degree~4 unavoidable sets we use typically have around 30 elements.
The probability that a given subset of size 30 of a sudoku solution grid does not contain five randomly chosen clues is
$$\frac{\binom{81 - 30}{5}}{\binom{81}{5}} \; = \; \frac{\binom{51}{5}}{\binom{81}{5}} \;\approx\; 0.092.$$
Therefore we expect that only about one in eleven unavoidable sets of degree 4 is \emph{not} hit after selecting five clues.
In other words, the state vector for the degree 4 unavoidable sets will carry 1's in approximately ten out of eleven places, on average.
(Recall that a 1 in the $i^{th}$ slot of this vector means that the $i^{th}$ unavoidable set of degree 4 has been hit already.)
For all future clues to be drawn, we only need to keep track of the unavoidable sets corresponding to the 0's in the state vector, since only those unavoidable sets have not yet been hit.
One way to do this is to recompute the 81~hitting vectors for this much smaller collection of degree~4 unavoidable sets, just as we did with the original collection right before the enumeration of hitting sets began.

A more efficient method is to take the original hitting vectors and shrink them to vectors of length at most 1,536, where the $i^{th}$ slot, \mbox{$0 \leq i \leq 32767$}, is erased according as the state vector carries a 1 or a 0 in this slot.
Pseudocode for the obvious algorithm to accomplish this --- which processes one slot at a time --- looks as follows:
\begin{verbatim}
var j := 0 // j is the index of an unavoidable set in the new collection
create array newquadhitvec[0..80]
for c from 0 to 80  // set all the new (shorter) hitting vectors to 0
  do newquadhitvec[c] := (0,...,0)
for i from 0 to 32767 // i is the index of a set in the original collection
  do if not quadstatevec[5].IsBitSet(i)
       then
         // the set of index i is not yet hit after picking five clues, so
         // it becomes the set of index j of the new collection
         // ==> for each clue, copy the i-th bit of the original hitting
         // vector to the j-th bit of the new hitting vector:
         for c from 0 to 80  
           do newquadhitvec[c].SetBit(j,quadhitvec[c].GetBit(i))
         inc(j)
         if j = 1536
           then return // we have found 1536 unhit unav'ble sets of degree 4
\end{verbatim}
However, it is possible to do considerably better than that, by processing more than one slot (bit) at a time.
For, suppose that we are given a consolidation (gathering) function
	$$\mathrm{Con}\colon\{0,1\}^8 \times \{0,1\}^8 \to \{0,1\}^8$$
that shrinks the second input vector such that precisely those slots are eliminated in which the first of the input vector carries a 1.
For example,
	$$\mathrm{Con}((0,1,1,0,1,0,1,0), (1,1,0,1,0,0,1,0)) \; = \; (1,1,0,0,0,0,0,0),$$
because the first input vector has 0's in slots 1, 4, 6 and 8 and the second vector has the digits 1, 1, 0 and 0, respectively, in these four slots, which therefore become the first four slots of the output vector; naturally slots 5--8 of the output are filled with 0's.
Suppose further that we are given a function
	$$\mathrm{hamwt}\colon\{0,1\}^8 \to \{0,\ldots,8\}$$
that computes the Hamming weight of a binary vector of length 8.
Then, given these two functions \mbox{$\mathrm{Con}$} and \mbox{$\mathrm{hamwt}$}, we can now easily process eight slots at a time:
\begin{verbatim}
var cnt := 0 // 'cnt' is the no. of unavoidable sets in the new collection, so far
var tmpvec   // a temporary binary vector of length 1536
for c from 0 to 80  
  do newquadhitvec[c] := (0,...,0)
for i from 0 to 4095 // 4096 = 32768 / 8
  do for c from 0 to 80  
    do tmpvec := Con(quadstatevec[5].GetByte(i),quadhitvec[c].GetByte(i))
       tmpvec.ShiftRight(cnt)
       newquadhitvec[c] := newquadhitvec[c] OR tmpvec
  cnt := cnt + (8 - hamwt(quadstatevec[5].GetByte(i)))
  if cnt >= 1536
    then return
\end{verbatim}
Here, the method $\mathrm{GetByte}$ is similar to $\mathrm{GetBit}$, except that it returns 8 bits at a time, while $\mathrm{ShiftRight}$ shifts the respective binary vector by the specified number of places, inserting 0's at the front.

\subsubsection*{Implementation notes}

In \emph{checker}, the function $\mathrm{Con}$ is implemented using the precomputed table $\mathrm{confunctab}$, while $\mathrm{hamwt}$ is already provided by the hardware through the $\mathrm{POPCNT}$ instruction.
So the above is really an implementation of 8-bit \emph{software} scatter.
Note that with the Haswell microarchitecture, Intel processors support (64-bit) \emph{hardware} scatter by way of the $\mathrm{PEXT}$ instruction.
Therefore, rewriting the procedure $\mathrm{ConsolidateHitvec}$ in \emph{checker} to take advantage of this instruction should result in a very significant performance increase on Haswell CPUs.

With the actual (Nehalem-optimized) code in \emph{checker}, we further performed partial loop unrolling with the inner of the above two loops, by always processing the hitting vectors for six clues in parallel.
Finally, one more slight improvement is gained by considering (shrinking) only those hitting vectors that correspond to clues that are not yet dead, and ignoring the other ones.

\subsection{Taking the effective size of the minimal unavoidable sets into account}\label{effective}

In this section we describe the final improvement over our original hitting set algorithm.
Recall that, at each stage, our hitting set algorithm adds clues from the first unavoidable set that is not yet hit, as described in Section~\ref{very_basic}.
Since the unavoidable sets are ordered by size, the set in question will always be one of smallest size.
However, this is not usually the best choice.
For instance, if the unavoidable set of lowest index that is not yet hit has empty intersection with the set of currently dead clues, and the unavoidable set of second lowest index that is not yet hit has the same size but one of its clues has been marked `dead' earlier, then it is obviously better to use the unavoidable set of second lowest index for drawing clues from.
For this reason, in the new \emph{checker}, when selecting the first ten clues we always use an unavoidable set of minimum effective size.\footnote{For the eleventh clue we still find the unavoidable set of minimum effective size among the first 64 unavoidable sets in our collection, and for the twelfth clue we find the unavoidable set of minimum effective size among the first five unavoidable sets not yet hit. For drawing the remaining four clues we always simply use the first unavoidable set not yet hit. The reasons for this are explained in Section~\ref{sec:checkerdetails}.}
Here, the \emph{effective size} of an unavoidable set is the number of clues it has that are not yet dead.

\subsubsection*{Implementation notes}

The way to efficiently accomplish this is to first invert the vector of dead clues, so that we obtain the \emph{vector of alive clues}, i.e., the binary vector that has a 1 in slot $i$ precisely if the clue $i$ is still alive.
Then, for each unavoidable set that is still unhit, we take the boolean $\mathrm{AND}$ of the vector of alive clues with the vector that has a 1 in exactly those slots corresponding to the clues this unavoidable set contains.
In other words, in the latter vector we simply set all slots to zero that correspond to clues that are dead.
We finally obtain the Hamming weight of the resulting vector, which is equal to the effective size of the unavoidable set in question.
We do this for all unavoidable sets, and we always remember the index of the first set that had the smallest effective size, so far.

\section{Running through all grids: the final computation}\label{sec:computation}

With the catalogue of all essentially different sudoku grids and \emph{checker} in hand, the actual search for 16-clue puzzles involved running \emph{checker} on each grid in the catalogue.
We outlined these steps in Section~\ref{summmethod}.
Because of the number of grids to check (about 5.5~billion) this had to be done using a large number of processors.
In this section we discuss the details of this computation.

\subsection{Some remarks on the implementation of the new checker}\label{sec:checkerdetails}

Following our successful PRACE prototype project application in 2009, we were able to test an early version of the new \emph{checker} on four different hardware platforms (AMD Istanbul, IBM Blue~Gene/P, IBM Power 6, Intel Nehalem), and we found that Nehalem is the best for us.
Hence we optimized \emph{checker} for Nehalem, to the extent of rewriting critical routines in assembly language.
We did this in such a way so as to maximize simultaneous use of the different execution units of a Nehalem core (instruction-level parallelism).
We further heavily used SSE to facilitate data-level parallelism.
Also, using machine language allowed us to retain key data \emph{checker} frequently uses inside the processor's registers, thereby further reducing overhead.
During the development of these assembly routines, we found the optimization manuals by Agner Fog \cite{fog} to be most helpful.
Moreover, we took advantage of the SMT mode of a Nehalem/Westmere CPU (hyper-threading), which was enabled on the cluster \emph{Stokes} at ICHEC that we used for this computation.
We give more details on the parallelisation later in this section.

Apart from the above, the main change to the implementation of the new \emph{checker} (compared to the version from 2006) is that we expanded all function calls when enumerating hitting sets.
So there are no recursive function calls when drawing more clues in the new \emph{checker}, rather, there are now sixteen nested loops, one for each clue.

One other implementation detail is that instead of performing a boolean $\mathrm{OR}$ (as described in the earlier sections) for updating the binary vectors that track the unavoidable sets of the various degrees, and later test if a vector is all 1's, we actually perform a boolean $\mathrm{AND}$ and check if the vector is all 0's.
The reason for this is that with SSE~4.2, it is slightly easier to test if an XMM-register (128 bits) is all 0's than if it is all 1's.

\subsubsection*{Tradeoffs}

During the development of \emph{checker}, there were several design choices that involved tradeoffs, i.e., there were often different alternatives that each had their advantages, so we had to find out the overall best (fastest) by experimentation.
These include:
\begin{itemize}
\item The number unavoidable sets to use, both minimal and of higher degree.
Obviously using more unavoidable sets results in fewer hitting sets being found, however, there is of course a price for having access to a larger selection of unavoidable sets.
We found that initially starting out with up to 384 minimal unavoidable sets was the optimum.
We also tried adding unavoidable sets of size 13 to \emph{checker}, but it seemed to make almost no difference to the average running time one way or another.
\item Which higher-degree unavoidable sets to use.
There does not seem to be a big difference whether or not one uses unavoidable sets of degree up to 5 only, or also unavoidable sets of degree~6.
However, adding unavoidable sets of degree~7 appears to slow down the computation.
\item The point at which we consolidate the hitting vectors.
Doing that later means a more powerful collection of unavoidable sets in the innermost loops of \emph{checker}, where most of the running time is spent.
However, doing that later also means having to do it more often.
The best combination could only be found through experimentation.
\item
Finding the best unavoidable set to use for drawing clues from.
In principal it would be best to always use the unavoidable set having the least effective size, but obviously not so if the performance hit incurred by finding this particular unavoidable set outweighs the gain.
Again, some experimentation was required.
\item
The number of chunks into which we divide a binary vector (of large length) when we test to see if every slot is 1.
Having more chunks means that we can potentially save work --- in case one of the slots of the vector is 0 --- but on the other hand it also means more $\mathrm{if}$-statements (to check whether or not a chunk is all 1's), i.e., more (mispredicted) branches.\footnote{In his document \emph{The microarchitecture of Intel, AMD and VIA CPUs}, Agner Fog states that according to his measurements, the penalty for a mispredicted branch on Nehalem is at least 17 clock cycles \cite{fog}.}
\end{itemize}

\subsection{Parallelisation strategy for the grid search}\label{sec:parallel}

The parallelisation strategy we used with this project was driven by the simple fact that each sudoku solution grid could be checked independently.
Therefore, as long as we ensured that all grids were actually processed, they could be handled in any order.
We focused our parallelisation effort on a HPC infrastructure. 
In this domain, the main parallelisation paradigm is MPI.
This \emph{Message Passing Interface} library gave us all the flexibility we needed to achieve an efficient and reliable parallelisation.

\subsubsection*{Loadbalancing and task-farming}

The final design of the parallel code was based on a master/slave architecture.
One group of master processes manages the reading of the grids to search, distributes the work to a pool of slave processes, collects the results, and finally writes them to disk.
The exact ratio between masters and slaves, along with the work distribution policy and the resulting printing pattern, was unclear at first.
That is why we made an initial version of the code very configurable in this regard, and we experimented with various options to find the most efficient one.
This took place during our access to the PRACE infrastructure for a prototype machine evaluation project.
Following our PRACE tests, the configuration we finally selected consists of one single master process per run.
This proved sufficient for the jobs we typically ran, which had less than 500~slave processes.
When using more than 500~slaves, running several independent jobs with one master process each proved more efficient at the batch scheduler level than increasing the size of the slaves' pool for one single job.

For distributing the work to the slave processes, the loadbalancing (which refers to the distribution, or \emph{task-farming}, of the tasks by a master node to the slave nodes) algorithm we implemented was as follows: a global list of sudoku grids to search was first read from disk by the master.
Next a (small) batch of grids was sent to each slave in a round-robin fashion.
Then the master would wait for a slave to contact it upon finishing.
Each slave process would run \emph{checker} on the batch of grids it had been assigned and then send the results back to the master.
The master would answer whichever slave tried to contact it first, collect its results, store the results on disk, check if more work was available, and either send new grids to search to the slave or command it to stop, according as sufficient (job) time was still available.
The master then would wait for the next slave to contact it, until all slaves had been commanded to stop.
Finally the master itself would stop and hence the job would terminate.
For optimum efficiency, we naturally wanted all slave nodes to finish around the same time.

This is the ideal case, and the job finishes properly.
However, it might happen that a job finishes unexpectedly.
For example, this would be the case if the batch scheduler killed the job for some reason, like having reached the wall-time limit.
In this case, we could restart the very same job, and all the previously processed sudoku grids would be automatically removed from the pool of work to do, avoiding any waste of resources.
Only grids that had either been partially processed by a slave, or finished but for which no acknowledgment had been transmitted to the master process, would be re-processed.
Because searching a grid usually takes only a few seconds, and a job duration is typically 24 to 84 hours, the potential number of grids to re-compute compared to the total number of grids checked in one compute job is tiny.
So the usage of the CPU cores was near-optimal, i.e., almost no processor time was lost at the end of a job.
Moreover, as the missing grids from a previous job would be the first ones to be computed at restart, if for some reason one grid takes an unusually long time to be searched by \emph{checker}, which increases the likelihood that the job will get killed for time-out during its computation, that grid would be re-processed right at the start of the next job and so there would be more than enough time to complete it.

The parallelisation technique we used was therefore well suited for our kind of workload, where the various individual jobs would take an unpredictable time to finish.
Having a first-come first-served work distribution policy ensured a natural load re-balancing between the slave processes, where the ones dealing with harder grids would simply report less often to the master process than the one with easier grids.

\subsubsection*{Extra performance refinements}

Selecting Intel Nehalem as our target platform gave us additional opportunity for performance tuning.
Indeed, this processor architecture supports a feature called SMT mode (hyper-threading).
This \emph{Simultaneous Multi-Threading} mode allows two processes to run concurrently on each CPU core.
Exploiting the SMT mode proved to be effective for us, leading to a 26.5 per cent performance gain at no cost.
To cope with the scheduling policy constraints on the \emph{Stokes} cluster, we had to further parallelise our code using OpenMP.
That way, we were able to pin each MPI process to one physical CPU core, where two OpenMP threads were running concurrently using the 2-way SMT mode available there.

\subsection{Ensuring the correctness of our programme checker}\label{sec:correctness}

Certainly the most important aspect of this project is correctness.
At first glance, it may appear difficult to verify the correctness of our programme \emph{checker} on the grounds of the lack of available test cases (no sudoku grids containing any 16-clue puzzles).
To be able to do some testing anyway, we produced a version of \emph{checker} that searches for 17-clue puzzles.
This version had minimal changes over the version that searches for 16-clue puzzles --- we made only those changes that were absolutely necessary.
We ran this \emph{checker} on all known grids having at least one 17-clue puzzle, and with every such grid, all the 17-clue puzzles that grid was known to contain were found by \emph{checker}.

As far as white-box testing goes, using the debugger we stepped through every line of the code, carefully verifying that each instruction does exactly what we thought it would do.
We also added pre- and post-conditions throughout the code (in the form of $\mathrm{assert}$-statements), to check the internal consistency of the data structures and to make sure that parameters passed were within their respective valid range.
For further debugging, we used the tool \emph{valgrind}, as well as its companion \emph{cachegrind}, the latter mainly for optimizing memory accesses.
Moreover, all changes to the code during development were tracked using a version control system.

We implemented two safety checks in \emph{checker}.
Firstly, after the procedure that finds all the minimal unavoidable sets in a grid is finished, we again test each set found, by running its complement through the solver, to make sure that the set is really unavoidable.
Secondly, we also double-check the answer the solver produces.
Each time a puzzle is run through the solver, we have the solver save the first two completions found.
Then we check that both completions are in fact valid completions of the given puzzle, and we further verify that they are actually different.
Should that ever not have been the case, the grid in question would have been logged.
However, this never actually happened, i.e., the solver always produced correct answers.

We technically used six different versions of \emph{checker} throughout the computations.
That is to say, we updated our code five times during the computations, where each version was a small improvement and a little faster than the previous version.
It would have been optimal to use a single version of the code for the entire computation, however, we were under time pressure due to the competitive nature of this project, and we needed to start running compute jobs as soon as was possible.
Had we used just the first version of the new \emph{checker} for the entire computation (as it was at the time we got access to the cluster \emph{Stokes}), then we would have needed at least 1.5 million additional compute hours on top of what we actually used.
Only through updating \emph{checker} were we able to finish the computation before the end of 2011.
Note however, that the difference between one version of \emph{checker} and the next one was usually just one procedure that was rewritten, i.e., the changes between updates were relatively minor, and of course we fully tested each version before the upgrades, as described in the beginning of this subsection.

\subsection{The actual computation}

The entire computation took about 7.1~million core hours on the \emph{Stokes} machine at ICHEC.
\emph{Stokes} is an SGI Altix ICE 8200EX cluster with 320 compute nodes. Each node has two Intel (Westmere) Xeon X5650 hex-core processors and 24~GB of RAM.
We divided the computation up into several hundred jobs.
We started running jobs in January 2011, and we finished in December 2011.
For each grid, we recorded the number of hitting sets found, and the time taken to search that grid. 
This data is stored in our log files.
The average running time of our final version of the new \emph{checker} is about $3.6$ seconds per grid (on a single core of the above model of CPU, with hyper-threading enabled).

\subsection{GPU computing opportunities}

The opportunity of using Graphic Processing Units (GPUs) for this project has been examined and rejected for the reason that at the time we could have begun porting the code to a GPU, we were already more than half-way through the computations.
Changing the code would have required development and validation time which, even with a dramatic speed-up, would not have permitted us to finish the project earlier.
Moreover, it is not clear to us if our hitting set algorithm is a suitable application for a GPU, at all.

For starters, one would need to come up with a very different parallelisation strategy --- exploiting the intrinsic parallel nature of the global workload is definitely not sufficient to perform an effective parallelisation at the GPU level.
The reason is that, due to the specific architecture of GPUs, a large number of threads would all need to work on a single grid, rather than on different grids.
This would necessitate parallelism at the code level itself, not just at the workload level.

At first, one might think this to be possible because checker mainly performs vector operations such as boolean-OR and reductions, which are well-suited for a GPU, provided the respective vectors are long enough.
Unfortunately, if the GPU is used only as a coprocessor for vector operations, then performance benefits are likely to be limited, because this approach would involve a lot of communications between the host and the GPU, i.e., the PCIe channel would almost certainly be a major bottleneck.
With that said, if larger chunks of the computations could be kept on the GPU, even if this means less (relative) efficiency there than on a CPU, then it might actually be possible to reap a performance gain. 
\section*{Acknowledgements}

This work has built on ideas and work of many other people.
It began from reading posts on the sudoku forums, and we thank many of the posters.
They include Guenter Stertenbrink, Gordon Royle, Ed Russell, Glenn Fowler, Roger Wanamo, who helped at various stages.
We also thank Konstantinos Drakakis and ICHEC for some extra CPU hours.
We thank the staff of ICHEC who were very supportive throughout the project.
Finally, we are grateful to Sascha Kurz for his comments on an earlier version of the manuscript.

\end{document}